\documentclass[conference]{IEEEtran}
\usepackage{graphicx}
\usepackage{amsthm}
\usepackage{epsfig}
\usepackage{latexsym}
\usepackage{amsfonts}
\usepackage{here}
\usepackage{rawfonts}
\usepackage[latin1]{inputenc}
\usepackage[T1]{fontenc}
\usepackage{calc}
\usepackage{capitalgreekitalic}
\usepackage{url}
\usepackage{enumerate}
\usepackage{color}
\usepackage[tbtags]{amsmath}
\usepackage{amssymb}
\usepackage{upref}
\usepackage{epic,eepic}
\usepackage{times}
\usepackage{dsfont}
\usepackage{comment}
\usepackage{cite}
\usepackage{bbm} 
\usepackage{amsmath}
\usepackage{microtype} 
\usepackage{afterpage}
\usepackage{makecell}

\newtheorem{lemma}{\bf Lemma}

\usepackage{dsfont}

\newcounter{step}
\newlength{\totlinewidth}
\newenvironment{algorithm}{%
  \rule{\linewidth}{1pt}
  \begin{list}{}%
    {\usecounter{step}%
      \settowidth{\labelwidth}{\textbf{Step 2:}}%
      \setlength{\leftmargin}{\labelwidth}%
      \setlength{\topsep}{-2pt}%
      \addtolength{\leftmargin}{\labelsep}%
      \addtolength{\leftmargin}{2mm}%
      \setlength{\rightmargin}{2mm}%
      \setlength{\totlinewidth}{\linewidth}%
      \addtolength{\totlinewidth}{\leftmargin}%
      \addtolength{\totlinewidth}{\rightmargin}%
      \setlength{\parsep}{0mm}%
      \raggedright}}%
  {\end{list}%
  \rule{\linewidth}{1pt}}
\newcounter{substep}

  {\end{list}}

\newlength{\aligntop}
\newlength{\alignbot}
 

\IEEEoverridecommandlockouts

\usepackage{algorithm}
\usepackage{algorithmic}

\makeatletter
\newcommand\semihuge{\@setfontsize\semihuge{19.3}{25}}
\makeatother

\makeatletter
\newcommand\semismall{\@setfontsize\semihuge{12.4}{15}}
\makeatother

\usepackage{subfigure}

\begin{document}

\title{\huge Optimizing Wireless Resource Management and Synchronization in Digital Twin Networks \vspace{-0.1cm}}

\author{\large{Hanzhi Yu, Yuchen Liu, \textit{Member IEEE}, Zhaohui Yang, \textit{Member IEEE}, } \\
\large{Haijian Sun, \textit{Member IEEE}, and Mingzhe Chen, \textit{Senior Member IEEE}\vspace{-0.2cm}} 

\thanks{Hanzhi Yu and Mingzhe Chen are with the Department of Electrical and Computer Engineering and Frost Institute for Data Science and Computing, University of Miami, Coral Gables, FL 33146 USA (Emails: \protect\url{hanzhiyu@miami.edu}; \protect\url{mingzhe.chen@miami.edu}).} 
\thanks{Yuchen Liu is with the Department of Computer Science, North Carolina State University, Raleigh, NC 27695 USA (Email: \protect\url{yuchen.liu@ncsu.edu}).} 
\thanks{Zhaohui Yang is with the College of Information Science and Electronic Engineering, Zhejiang University, Hangzhou 310027, China (Email: \protect\url{yang_zhaohui@zju.edu.cn}).}
\thanks{Haijian Sun is with the School of Electrical and Computer Engineering, University of Georgia, Athens, GA 30602 USA (Email: \protect\url{hsun@uga.edu}).}
\thanks{This work was supported by the U.S. National Science Foundation under Grants CNS-2350076, CNS-2312139, CNS--2312138, and SaTC--2350075. }

}
\maketitle
%
\begin{abstract}
In this paper, we investigate an accurate synchronization between a physical network and its digital network twin (DNT), which serves as a virtual representation of the physical network. The considered network includes a set of base stations (BSs) that must allocate its limited spectrum resources to serve a set of users while also transmitting its partially observed physical network information to a cloud server to generate the DNT. Since the DNT can predict the physical network status based on its historical status, the BSs may not need to send their physical network information at each time slot, allowing them to conserve spectrum resources to serve the users. However, if the DNT does not receive the physical network information of the BSs over a large time period, the DNT's accuracy in representing the physical network may degrade. To this end, each BS must decide when to send the physical network information to the cloud server to update the DNT, while also determining the spectrum resource allocation policy for both DNT synchronization and serving the users. We formulate this resource allocation task as an optimization problem, aiming to maximize the total data rate of all users while minimizing the asynchronization between the physical network and the DNT. The formulated problem is challenging to solve by traditional optimization methods, as each BS can only observe a partial physical network, making it difficult to find an optimal spectrum allocation strategy for the entire network. To address this problem, we propose a method based on the gated recurrent units (GRUs) and the value decomposition network (VDN). The GRU component allows the DNT to predict future status using the historical data, effectively updating itself when the BSs do not transmit the physical network information. The VDN algorithm enables each BS to learn the relationship between its local observation and the team reward of all BSs, allowing it to collaborate with others in determining whether to transmit physical network information and optimizing spectrum allocation. Simulation results show that our GRU and VDN based algorithm improves the weighted sum of data rates and the similarity between the status of the DNT and the physical network by up to 28.96\%, compared to a baseline method combining GRU with the independent Q learning (IQL).
\end{abstract}


\section{Introduction}\label{se:intro}
Digital twin (DT) is a virtual representation of a physical product or process, used to understand and predict the physical counterpart's performance characteristics \cite{9429703, 9839640, liu2024digital}. Different from those traditional simulation tools, which use computer-based models or mathematical concepts to test systems, processes, and the effects of various variables, a DT utilizes real-time data from its associated physical object for simulations, analysis, and online control. The two-way information flow improves the performance of the predictive analytical model \cite{9854866}. Based on the definition of the DT, the concept of the DT network (DTN), or the digital network twin (DNT), is generated to describe a network that constructed by multiple one-to-one DTs. A DNT uses advanced communication technologies to realize information sharing not only between each physical object and its twin, but also among different physical objects and among different twins \cite{9429703, 10012285}. According to previous works, the creation of a DNT presents several challenges. First, constructing a DNT requires mapping not only physical objects but also several unique networking factors (i.e., network protocols, wireless channel dynamics, and the network performance metrics). Hence, it is impractical to directly map all network features for DNT generation and one must select appropriate network features for DNT creations \cite{zhang2024mapping}. Second, there is no standardized metrics exist for evaluating the DNT synchronization \cite{sharma2022digital}. Third, privacy protection in DTNs is a crucial issue, as they are vulnerable to data breaches and malicious attacks during information exchange, necessitating robust security measures to protect users' private information. \cite{9429703, 10148936}.

A number of existing works \cite{8764584, 10234596, 10148936, zhou2022digital} have studied the generation and creations of DNTs. In particular, the authors in \cite{10148936} presented several use cases, the design standardization, and an implementation example of the DNT. In \cite{8764584}, the authors designed a deep neural network (DNN) based method to create a DNT for the approximation of the optimal user association scheme in a mobile edge computing network. The work in \cite{10234596} created a DNT using a Bayesian model. In \cite{zhou2022digital}, the authors proposed a DT empowered framework for optimizing network resource management, where DTs collect real-time data from users to predict and dynamically allocate network resources, in order to enhance efficiency of resource utilization and reduce reconfiguration costs. However, most of these works \cite{8764584, 10234596, 10148936, zhou2022digital} did not consider the similarity and synchronization between the DNT and a physical network, and assumed that the DNT is already synchronized with a physical network without considering how the DNT collects physical network data to achieve synchronization. 

To address this issue, a number of existing studies \cite{hashash2022edge, guo2023resource, zheng2023data, 9865226} have focused on the DT synchronization optimization. In particular, the authors in \cite{hashash2022edge} proposed a continual learning framework to build a synchronized DNT of a autonomous vehicle network so as to make vehicle driving decisions. In \cite{guo2023resource}, the authors proposed a DT-empowered framework for resource allocation in UAV-assisted edge mobile networks, leveraging a deviation model to address discrepancies between the DT and physical states. In \cite{zheng2023data}, the authors aimed to optimize data synchronization in vehicular DT networks by developing a deep learning model that predicts relay waiting times. In \cite{9865226}, the authors designed a dynamic hierarchical synchronization framework for IoT-assisted DTs in the Metaverse, and the framework optimizes synchronization intensities through a multi-level game-theoretic approach. However, these works 
\cite{hashash2022edge, guo2023resource, zheng2023data, 9865226} did not consider how the generation of DNTs affect the physical network performance since DNT generations require the transmission of a large amount of data, which will also introduce significant communication overhead. 

The main contribution of this work is to design a novel DNT framework that jointly optimizes the performance of a physical network and the synchronization between the DNT and the physical network. Our key contributions include: 
\begin{itemize}
    \item We consider a DNT enabled network that consists of a physical network with a set of base stations (BSs), several users, and a DNT. The DNT is a virtual representation of the physical network and can predict physical network dynamics. Each BS must use limited spectrum resources to serve the users and transmit the information of the physical network to a cloud server to generate the DNT. Since the DNT can predict the physical network status, the BSs may not need to transmit the information to the server at every time slot, and thus conserving spectrum resources to better serve the users. To this end, each BS in the physical network needs to determine whether to transmit the physical network information to the cloud server for updating the DNT, and optimize spectrum resource allocation for the users and physical network information transmission. We formulate this problem as an optimization problem aiming to maximize the data rates of all users while minimizing the gap between the states of the physical network and the DNT. 
    \item The formulated problem is challenging to solve since each BS cannot fully observe the entire status of the physical network. To solve this problem, we proposed a machine learning (ML) method that integrates the gate recurrent units (GRUs) \cite{cho2014learning, chung2014empirical} and the value-decomposition based reinforcement learning (VD-RL) method \cite{sunehag2017value}. The GRUs allow the DNT to predict its future status using historical data and to update its status when the DNT cannot receive the information of the physical network. The VD-RL method can learn from the GRU prediction accuracy to enable each BS to decide its associated users, whether to send the physical network information to the cloud server, and RB allocation, optimizing the data rate of all users in the physical network and ensuring an accurate synchronization between the physical network and the DNT. Compared to other RL methods \cite{tan1993multi}, the VD-RL method allows each BS to use its partial observation, specifically, the locations of the users in its coverage, to collaboratively find a globally optimal solution for both physical network and DNT. 
\end{itemize}
Simulation results show that our proposed method improves the weighted sum of data rates and the the similarity between the status of the DNT and the physical network by up to 28.96\%, compared to a baseline method combining GRU with the independent Q learning (IQL).

The rest of this paper is organized as follows. The system model and problem formulation are introduced in Section \ref{se:system}. The design of the GRU and VD-RL integrated method will be introduced in Section \ref{se:solution}. The analysis of the complexity, the convergence, and the implementation of the designed method are studied in Section \ref{se:analysis}. Simulation results are presented and discussed in Section \ref{se:simulation}. Finally, conclusions are drawn in Section \ref{se:conclusion}. 

\section{System Model and Problem Formulation}\label{se:system}
We consider a DNT enabled network which consists of: 1) a physical network including a set $\mathcal{U}$ of $U$ mobile users and a set $\mathcal{M}$ of $M$ BSs, and 2) a DNT that is generated and controlled by a cloud server, as shown in Fig. \ref{fig:sys}. Each BS can provide communication service to the users that are located in its coverage. Meanwhile, each BS must transmit the physical network information collected from its associated users to the cloud server for the generation of the DNT. We first introduce the mobility model of each user. Next, we introduce the transmission model. Then, we describe the DNT model. Finally, we present our considered optimization problem. Table \ref{tab:notation} provides a summary of the notations used hereinafter. 
\begin{figure}[!t]
  \begin{center}
    \includegraphics[width=8cm]{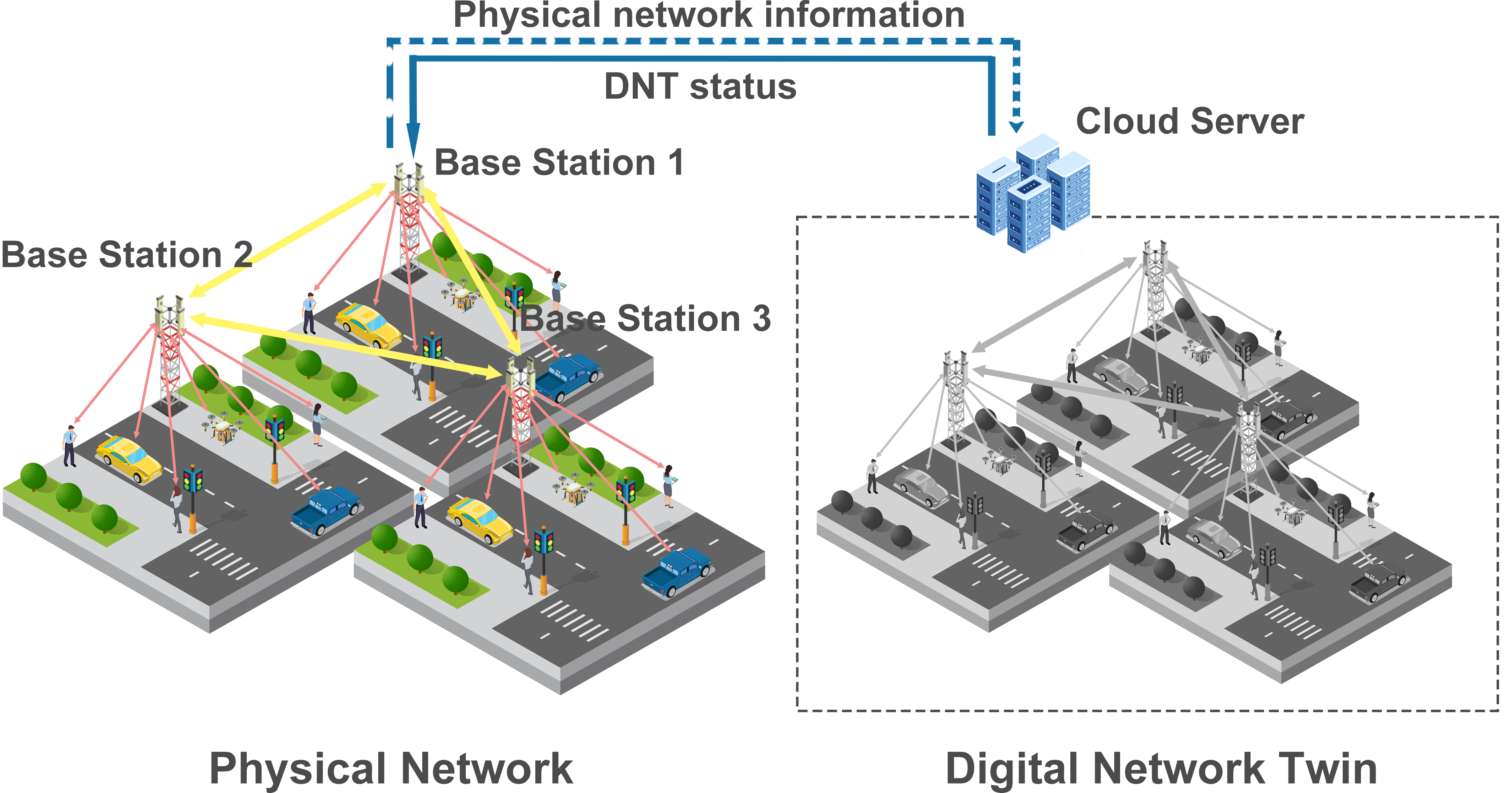}
    \caption{The considered DNT enabled network. } 
    \label{fig:sys}
  \end{center}
\end{figure}

\begin{table*}
\renewcommand\arraystretch{1.1}
\caption{List of Notations}
\label{tab:notation}
\centering
\begin{tabular}{|c|c|c|c|}
\hline
\textbf{Notation} & \textbf{Description} & \textbf{Notation} & \textbf{Description}\\
\hline

$N$ & Number of RBs of each BS & $U$ & Number of users \\
\hline
$\mathcal{N}$ & Set of RBs & $\mathcal{U}$ & Set of users \\
\hline
$\boldsymbol{p}_u$ & Probability vector of user $u$ chooses each possible movement & $\boldsymbol{l}^\textrm{U}_{u,t}$ & The position of user $u$ at time slot $t$ \\
\hline
$c_{umt}$ & Data rate of BS $m$ transmitting data to user $u$ at time slot $t$ & $\Delta l$ & Distance a user can move in one time slot \\
\hline
$d_{umt}$ & Distance between user $u$ and BS $m$ at time slot $t$ & $\boldsymbol{x}_{umt}, \boldsymbol{y}_{mt}$ & RB allocation vector \\
\hline
$I^\textrm{U}_{umt,n}, I^\textrm{C}_{mt,n}$ & Interference caused by other BSs that use RB $n$ & $B$ & Bandwidth of each RB \\
\hline
$c^\textrm{C}_{mt}$ & Data rate of BS $m$ transmitting data to the cloud server & $P$ & Transmission power \\
\hline
$h_m \left( \boldsymbol{l}_{u,t}^\textrm{U} \right)$ & Channel gain between user $u$ and BS $m$ at time slot $t$ & $o_u$ & Rayleygh fading parameter \\
\hline
$h_m \left( \boldsymbol{l}^\textrm{C} \right)$ & Channel gain between BS $m$ and the cloud server & $\boldsymbol{l}^\textrm{C}$ & The position of the cloud server \\
\hline
$D_m$ & The size of data needed to be transmitted to the cloud server & $\boldsymbol{l}_m^\textrm{M}$ & The position of BS $m$ \\
\hline
$T_{mt}$ & Transmission delay of BS $m$ transmitting data to the cloud server & $\boldsymbol{s}_t$ & Physical network status at time slot $t$ \\
\hline
$\boldsymbol{s}^m_t$ & Partial physical network status observed by BS $m$ & $\overline{\boldsymbol{s}}^m_t$ & A mapping of $\boldsymbol{s}^m_t$ \\
\hline
$\boldsymbol{\hat s}^m_t$ & The status of BS $m$ estimated by the DNT & $\epsilon$ & Weight parameter \\
\hline
$\boldsymbol{h}_t$ & Hidden states of the GRU at time slot $t$ & $\alpha$ & Threshold of the transmission delay \\
\hline
\makecell[c]{$\boldsymbol{W}^{\textrm{r}}, \boldsymbol{W}^{\overline{\textrm{h}}}, $ \\ $\boldsymbol{W}^{\textrm{z}}, \boldsymbol{W}^{\textrm{o}}$} & Weight matrix of the GRU & $\boldsymbol{U}^{\textrm{r}}, \boldsymbol{U}^{\overline{\textrm{h}}}, \boldsymbol{U}^{\textrm{z}}$ & Weight matrix of the GRU \\
\hline
$N^{\textrm{h}}$ & Number of units in the hidden layer of the GRU & $\boldsymbol{r}^{\textrm{G}}_t$ & The reset gate of the GRU \\
\hline
$K$ & Length of the GRU input sequence & $\tilde{\boldsymbol{h}}_{\tau}$ & Candidate hidden state \\
\hline
$\boldsymbol{z}_{mt}$ & Association users vector of BS $m$ at time slot $t$ & $\boldsymbol{z}^{\textrm{G}}_{\tau}$ & The update date of the GRU \\
\hline
$\lambda_G$ & Learning rate of the GRU & $\boldsymbol{a}^m_t$ & The action of BS $m$ at time slot $t$ \\
\hline
$\boldsymbol{a}_t$ & The joint action of all BSs at time slot $t$ & $\boldsymbol{\pi}^m \left( \boldsymbol{a}^m_t | \boldsymbol{s}^m_t \right)$ & The policy of BS $m$ \\
\hline
$\xi_{ut}$ & Number of BSs that serve user $u$ at time slot $t$ & $r \left( \boldsymbol{s}_t, \boldsymbol{a}_t \right)$ & The team reward \\
\hline
$\rho$ & Penalty for one user obtaining RBs from multiple BSs & $\mathcal{G}^m$ & Bipartite graph of BS $m$ \\
\hline
$\mathcal{U}^m$ & Set of vertices represent the associated users of BS $m$ & $\mathcal{N}^m$ & Set of vertices represent the RBs of BS $m$ \\
\hline
$\mathcal{E}^m$ & Set of edges that connect vertices from $\mathcal{U}^m$ and $\mathcal{N}^m$ & $e^m_{un}$ & The edge in $\mathcal{E}^m$ \\
\hline
$\mathcal{X}^m_t$ & A subset of $\mathcal{E}^m$ & $\psi^m_{un}$ & The weight of edge $e_{un}^m$ \\
\hline
$\boldsymbol{\theta}_m$ & The parameters of BS $m$'s Q network & $\boldsymbol{\tilde{\theta}}_m$ & The parameters of BS $m$'s target network \\
\hline
$\mathcal{D}_g$ & Set of transitions used to train the Q networks in epoch $g$ & $G$ & Number of the Q networks training epoch \\
\hline
$D$ & Transitions collected in one training epoch & $\lambda_Q$ & Learning rate of the Q networks \\
\hline
$\theta^{\textrm{h}}$ & Size of the hiddent states of the Q networks & & \\
\hline
\end{tabular}
\end{table*}

\subsection{Mobility Model}
For simplicity, we use a random walk model to describe the mobility of each user \cite{tabassum2019fundamentals}. At each time slot $t$, each user $u$ can choose from five possible movements: 1) stay at the current location, 2) move forward, 3) move backward, 4) move left, and 5) move right. The probability of each user $u$ choosing each possible movement is $\boldsymbol{p}_u = \left[ p_{u,1}, p_{u,2}, p_{u,3}, p_{u,4}, p_{u,5} \right]$. We assume that the position of user $u$ at time slot $t$ is $\boldsymbol{l}^{\textrm{U}}_{u,t} = \left[ l^{\textrm{U}}_{u,t,1}, l^{\textrm{U}}_{u,t,2} \right]$. Thus, the position of user $u$ at time slot $t+1$ is 
\begin{equation}\label{eq:move}
     \boldsymbol{l}_{u, t+1}^\textrm{U} = \left \{
        \begin{aligned}
            &\left[ l_{u,t,1}^\textrm{U}, l_{u,t,2}^\textrm{U} \right] &&\text{with probability } p_{u,1}, \\
            &\left[ l_{u,t,1}^\textrm{U}, l_{u,t,2}^\textrm{U} + \Delta l \right] &&\text{with probability } p_{u,2}, \\
            &\left[ l_{u,t,1}^\textrm{U}, l_{u,t,2}^\textrm{U} - \Delta l \right] &&\text{with probability } p_{u,3}, \\
            &\left[ l_{u,t,1}^\textrm{U} - \Delta l, l_{u,t,2}^\textrm{U} \right] &&\text{with probability } p_{u,4}, \\
            &\left[ l_{u,t,1}^\textrm{U} + \Delta l, l_{u,t,2}^\textrm{U} \right] &&\text{with probability } p_{u,5}, \\
        \end{aligned}
    \right.
\end{equation}
with $\Delta l$ being the distance each user can move in one time slot. 

\subsection{Transmission Model}
In our considered network, each BS $m$ utilizes an orthogonal frequency division multiple access (OFDMA) technique to serve its associated users and transmit the information of the physical network over a set $\mathcal{N}$ of $N$ resource blocks (RBs). Each user can only be served by one BS and each BS will only allocate one RB to a user at each time slot. The transmission rate of BS $m$ transmitting data to user $u$ at time slot $t$ is 
\begin{equation}\label{eq:c}
\begin{split}
    c_{umt} &\left( \boldsymbol{x}_{umt}, \boldsymbol{X}_{\left( -m \right) t}, \boldsymbol{Y}_{\left( -m \right) t} \right) \\
    & = \sum_{n=1}^N x_{umt,n} B \log_2 \left( 1 + \frac{P h_m \left( \boldsymbol{l}_{u,t}^\textrm{U} \right)}{I^\textrm{U}_{umt,n} + B N_0} \right), 
\end{split}
\end{equation}
where 
$\boldsymbol{X}_{\left( -m \right) t} = \left[ \boldsymbol{x}_{ujt} \right]_{u \in \mathcal{U}, j \in \mathcal{M}, j \neq m}$, $\boldsymbol{x}_{umt} = \left[ x_{umt,1}, ..., x_{umt,N} \right]$ is an RB allocation vector, with $x_{umt,n} \in \{ 0,1 \}$ indicating whether RB $n$ of the BS $m$ is allocated to user $u$ at time slot $t$; 
$\boldsymbol{Y}_{\left( -m \right) t}  = \left[ \boldsymbol{y}_{1t}, \ldots, \boldsymbol{y}_{\left( m-1 \right)t}, \boldsymbol{y}_{\left( m+1 \right)t},\ldots, \boldsymbol{y}_{Mt} \right]$ with $\boldsymbol{y}_{mt} = \left[ y_{mt,1}, ..., y_{mt,N} \right]$ being the RB allocation vector of BS $m$ for the physical network information transmission, with $y_{mt,n} \in \{ 0,1 \}$ indicating whether RB $n$ of BS $m$ is allocated for the physical network information transmission;
$B$ is the bandwidth of each RB; $P$ is the transmit power; $h_m \left( \boldsymbol{l}_{u,t}^\textrm{U} \right) = o_u d_{umt}^{-2}$ is the channel gain between user $u$ and BS $m$, with $o_u$ being the Rayleigh fading parameter, $d_{umt} = \sqrt{\lVert \boldsymbol{l}_{u,t}^\textrm{U} - \boldsymbol{l}_m^\textrm{M} \rVert_2}$ being the distance between user $u$ and BS $m$ at time slot $t$, $\boldsymbol{l}_m^\textrm{M}$ being the position of BS $m$; $N_0$ is the noise power spectral density; $h_j \left( \boldsymbol{l}^\textrm{C} \right)$ is the channel gain between BS $j$ and the cloud server, with $\boldsymbol{l}^\textrm{C}$ being the position of the cloud server; and $I^\textrm{U}_{umt,n}$ is the interference caused by other BSs that use RB $n$ for serving users and physical network information transmissions, which is given by 
\begin{equation}
\begin{split}
    I^\textrm{U}_{umt,n} = \sum_{i \in \mathcal{U}, i \neq u} \sum_{j \in \mathcal{M}, j \neq m} &\left(x_{ijt,n} P h_j \left( \boldsymbol{l}^\textrm{U}_{u,t} \right) \right. \\
    &\left. + y_{jt,n} P h_j \left( \boldsymbol{l}^\textrm{U}_{u,t} \right) \right). 
\end{split}
\end{equation}
Similarly, the transmission rate of BS $m$ transmitting physical network information to the cloud server at time slot $t$ is 
\begin{equation}\label{eq:i_u}
\begin{split}
    c^\textrm{C}_{mt} &\left( \boldsymbol{y}_{mt}, \boldsymbol{X}_{\left( -m \right) t}, \boldsymbol{Y}_{\left( -m \right) t} \right) = \\
    &\sum_{n=1}^N y_{mt,n} B \log_2 \left( 1 + \frac{P h_m \left( \boldsymbol{l}^\textrm{C} \right)}{I^\textrm{C}_{mt,n} + B N_0} \right), 
\end{split}
\end{equation}
where $I^\textrm{C}_{mt,n}$ is the interference caused by other BSs that use RB $n$ for serving users and physical network information transmission, which is given by
\begin{equation}\label{eq:i_c}
    I^\textrm{C}_{mt,n} = \sum_{i \in \mathcal{U}} \sum_{j \in \mathcal{M}, j \neq m} \left( x_{ijt,n} P h_j \left( \boldsymbol{l}^\textrm{C} \right) + y_{jt,n} P h_j \left( \boldsymbol{l}^\textrm{C} \right) \right). 
\end{equation}
We assume that the data size of the physical network information that is needed to be transmitted to the cloud server is $D_m$. Given the data rate $c^\textrm{C}_{mt}$, the transmission delay of BS $m$ transmitting physical network information to the cloud server at time slot $t$ can be represented as 
\begin{equation}
    T_{mt} = \frac{D_m}{c^\textrm{C}_{mt} \left( \boldsymbol{y}_{mt}, \boldsymbol{X}_{\left( -m \right) t}, \boldsymbol{Y}_{\left( -m \right) t} \right)}. 
\end{equation}

\subsection{Digital Network Twin Model}
As a virtual representation of the physical network, the DNT should have the same status and the network management strategy with the physical network, even when some BSs do not transmit the physical network information. We first define a vector $\boldsymbol{s}_t$ to represent the status of the physical network at time slot $t$. In our considered network, the positions of all $U$ users are used to describe the status of the physical network. Since each BS has its service area, it can only serve a subset $\mathcal{U}^m_t$ of users that are located in its coverage. Thus, the entire physical network status $\boldsymbol{s}_t$ is partially observed by each BS $m$. We use $\boldsymbol{s}^m_t = \left[ \boldsymbol{l}_{u,t}^\textrm{U} \right]_{u \in \mathcal{U}^m_t}$ to denote physical network status observed by BS $m$ at time slot $t$. Then, we have $\boldsymbol{s}_t = \left[ \boldsymbol{s}^1_t, ..., \boldsymbol{s}^M_t \right]$. Accordingly, the DNT status at time slot $t$ can be defined as $\overline{\boldsymbol{s}}_t = \left[ \overline{\boldsymbol{s}}^1_t, ..., \overline{\boldsymbol{s}}^M_t \right]$, with $\overline{\boldsymbol{s}}^m_t$ being a mapping of the physical network status $\boldsymbol{s}^m_t$. Given $\boldsymbol{y}_{mt}$ and $\boldsymbol{s}^m_t$, $ \overline{\boldsymbol{s}}^m_t$ is given by 
\begin{equation}\label{eq:s_dnt}
     \overline{\boldsymbol{s}}^m_t\left( \boldsymbol{y}_{mt} \right)= \left \{
        \begin{aligned}
            &\boldsymbol{s}^m_t, &&\text{if}~\sum_{n=1}^N y_{mt,n} = 1,\\
            &\boldsymbol{\hat s}^m_t &&\text{if}~\sum_{n=1}^N y_{mt,n} = 0, \\
        \end{aligned}
    \right.
\end{equation}
where $\boldsymbol{\hat s}^m_t = \left[ \hat{\boldsymbol{l}}_{u,t}^\textrm{U} \right]_{u \in \mathcal{U}^m_t}$ is the status estimated by the DNT at time slot $t$. From (\ref{eq:s_dnt}), we can see that if each BS $m$ allocates an RB to transmit its partially observed physical network information to the cloud server (i.e., $\sum_{n=1}^N y_{mt,n} = 1$), the status of the DNT and the status of the physical network is similar, since the DNT can obtain the physical network status from BS $m$ directly. Otherwise, if BS $m$ does not allocate any RBs to the cloud server, the DNT must estimate the physical network status of BS $m$. When the physical network status and the DNT status are identical (i.e., $\boldsymbol{s}_t = \overline{\boldsymbol{s}}_t$), we consider the DNT synchronizes with the physical network. 

\subsection{Problem Formulation}
Given the our designed system model, we next describe our considered optimization problem. Our goal is to maximize the sum of the data rates of all $U$ users, while guaranteeing the synchronization between the DNT and the physical network over a set $\mathcal{T}$ of $T$ time slots. 
The optimization problem includes optimizing the RB allocation matrix $\boldsymbol{X}_{t}$ and physical network information transmission matrix $\boldsymbol{y}_{t}$. The optimization problem can be formulated as 
\begin{subequations}\label{eq:problem}
    \begin{equation}\tag{\theequation}\label{eq:opt_prob}
    \begin{split}
        \max_{\{ \boldsymbol{X}_t, \boldsymbol{Y}_{t} \}_{t \in \mathcal{T}}} & \sum_{t=1}^T \left( - \frac{1 - \epsilon}{U} \lVert \boldsymbol{s}_t - \overline{\boldsymbol{s}}_t \rVert_2^2 \right. \\
        & \left. + \epsilon \sum_{m=1}^M \sum_{u=1}^U c_{umt} \left( \boldsymbol{x}_{umt}, \boldsymbol{X}_{\left( -m \right) t}, \boldsymbol{Y}_{\left( -m \right) t} \right) \right), 
    \end{split}
    \end{equation}
    \begin{flalign}
        &&\text{s.t. } &
        x_{umt,n} \in \{ 0,1 \}, m \in \mathcal{M}, n \in \mathcal{N}, u \in \mathcal{U}, t \in \mathcal{T}, && \label{eq:const1} \\
        && & \sum_{m=1}^M \sum_{n=1}^N x_{umt,n} \leq 1, u \in \mathcal{U}, t \in \mathcal{T}, && \label{eq:const2} \\
        && & \sum_{u=1}^U x_{umt,n} \leq 1, m \in \mathcal{M}, n \in \mathcal{N}, t \in \mathcal{T}, && \label{eq:const3} \\
        && & y_{mt,n} \in \{ 0,1 \}, m \in \mathcal{M}, n \in \mathcal{N}, t \in \mathcal{T}, && \label{eq:const4} \\
        && &\sum_{n=1}^N y_{mt,n} \leq 1, m \in \mathcal{M}, t \in \mathcal{T}, && \label{eq:const5} \\
        && &\sum_{u=1}^U x_{umt,n} + y_{mt,n} \leq 1, m \in \mathcal{M}, n \in \mathcal{N}, t \in \mathcal{T}, && \label{eq:const6} \\
        && & T_{mt} \leq \alpha, m \in \mathcal{M}, t \in \mathcal{T}, && \label{eq:const7} 
    \end{flalign}
\end{subequations}
where $\epsilon \in \left( 0,1 \right)$ is a weight parameter, and $\alpha$ is a threshold of the transmission delay of each BS transmitting the physical network information to the cloud server. In problem (\ref{eq:problem}), constraints (\ref{eq:const1}) and (\ref{eq:const2}) imply that each user can be served by only one BS and can occupy only one RB of this BS. Constraints (\ref{eq:const1}) and (\ref{eq:const3}) implies that each RB can be occupied by only one user. Constraints (\ref{eq:const4}) and (\ref{eq:const5}) imply that each BS $m$ can allocate only one RB for physical network information transmission. (\ref{eq:const6}) implies that each RB $n$ of BS $m$ can be used either for serving a user or for physical network information transmission. Constraint (\ref{eq:const7}) is a constraint of the delay of each BS $m$ transmitting physical network information to the cloud server.

Problem (\ref{eq:problem}) is challenging to solve by conventional optimization methods due to the following reasons. First, each BS $m$ cannot fully observe the entire physical network status $\boldsymbol{s}_t$, making it difficult to find the RB allocation policy which optimizes the data rates of all users and the DNT synchronization using conventional optimization methods \cite{lovejoy1991survey}. Second, each BS cannot know the future DNT status as it depends on the observations of the physical network, the decisions of physical network information transmission by other BSs, as well as the estimation accuracy of the DNT. Third, our considered system is dynamic due to user movements. Hence, the status of the physical network and the DNT will change over time. This makes it complicated to solve problem (\ref{eq:problem}) using conventional optimization methods because we need to resolve the problem when the physical network statuses change. To solve problem (\ref{eq:problem}), we propose a value decomposition based MARL method which enables each BS to find its optimal RB allocation policy and synchronization policy, while considering the RB allocation and synchronization policy of others, thus optimizing the data rate of all users in the physical network and keeping an accurate synchronization between the physical network and the DNT. 

\section{Solution of Problem}\label{se:solution}
To solve problem (\ref{eq:problem}), we next introduce our proposed ML based method that combines GRUs with the value decomposition network (VDN). In our designed method, the GRU is implemented by the DNT to estimate the status of the physical network. The VDN is used by each BS to determine its associated users, as well as allocating RBs to serve the associated users and to transmit physical network information to the cloud server. Meanwhile, the use of GRUs enables the designed method to use historical user mobility data for future user movement prediction without the reliance on user mobility models. Compared to traditional recurrent neural networks (RNNs) methods for user mobility prediction, the GRUs can effectively fix the gradient vanishing problem thus improving prediction accuracy and efficiency. Compared to traditional multi-agent reinforcement learning (MARL) methods (i.e., IQL \cite{tan1993multi}), the VDN allows each BS to decide its associated users and whether to transmit the physical network information based on its partial observation, but optimize the performance of the entire system (i.e., the sum of data rates of all users, and the synchronization between the DNT and the physical network). 
Next, we first introduce the GRU for physical network status estimation. Then, we introduce the components of the VDN based MARL framework. Finally, we explain the procedure of using our proposed method to solve problem (\ref{eq:problem}). 

\subsection{GRUs for Physical Network Status Estimation} 
We first introduce the GRUs for physical network status estimation. The GRU-based estimation model is deployed at the cloud server, and its output is an estimated status of the physical network of all BSs. Accurate status estimation enables the DNT to simulate the physical network, even when the BSs do not transmit their partially observed physical network information to the cloud server. The GRU-based estimation model consists of three components: 1) input, 2) output, and 3) the GRU model, which are introduced as follows. 

\subsubsection{Input} The input of the GRU-based estimation model consists of the most recent $K$ statuses of the entire DNT, which can be represented as $\overline{\boldsymbol{S}}_t = \left[ \overline{\boldsymbol{s}}_{t-K+1}, ..., \overline{\boldsymbol{s}}_{t} \right]$. 

\subsubsection{Output} The output of the GRU-based estimation model is the estimated status $\hat{\boldsymbol{s}}_{t+1}$. 

\subsubsection{The GRU Model} The GRU model is used to approximate the relationship between the input $\overline{\boldsymbol{S}}_t$ and the output $\hat{\boldsymbol{s}}_{t+1}$. A GRU model consists of an input layer, a hidden layer, and an output layer. The hidden states $\boldsymbol{h}_t$ of the units in the hidden layer at time slot $t$ are used to store the information related to the previous states from time slots $t-K+1$ to $t$. At each time slot $\tau$ from $t-K+1$ to $t$, the hidden states $\boldsymbol{h}_{\tau}$ is determined by the reset gate $\boldsymbol{r}^{\textrm{G}}_{\tau}$ and the update gate $\boldsymbol{z}^{\textrm{G}}_{\tau}$. The reset gate $\boldsymbol{r}^{\textrm{G}}_{\tau}$ determines how much of the previous hidden states $\boldsymbol{h}_{\tau-1}$ should influence the new candidate state, effectively allowing the model to retain relevant portions of past information. The reset gate at time slot $\tau$ is  
\begin{equation}\label{eq:reset}
    \boldsymbol{r}^{\textrm{G}}_{\tau} = \phi \left( \boldsymbol{W}^{\textrm{r}} \overline{\boldsymbol{s}}_{\tau} + \boldsymbol{U}^{\textrm{r}} \boldsymbol{h}_{\tau-1} \right), 
\end{equation}
where $\phi \left( \cdot \right)$ is the sigmoid function, 
and $\boldsymbol{W}^{\textrm{r}} \in \mathbb{R}^{N^{\textrm{h}} \times 2U}$ and $\boldsymbol{U}^{\textrm{r}} \in \mathbb{R}^{N^{\textrm{h}} \times N^{\textrm{h}}}$ are the weight matrices of the reset gate with $N^{\textrm{h}}$ being the number of the units in the hidden layer. Given the reset gate $\boldsymbol{r}^\textrm{G}_{\tau}$, the new candidate hidden state $\tilde{\boldsymbol{h}}_{\tau}$ is 
\begin{equation}\label{eq:candidate}
    \tilde{\boldsymbol{h}}_{\tau} = \textrm{tanh} \left( \boldsymbol{W}^{\tilde{h}} \overline{\boldsymbol{s}}_{\tau} + \boldsymbol{U}^{\tilde{h}} \left( \boldsymbol{h}_{\tau-1} \odot \boldsymbol{r}^{\textrm{G}}_{\tau} \right) \right), 
\end{equation}
where $\textrm{tanh} \left( \cdot \right)$ is the hyperbolic tangent function, $\boldsymbol{W}^{\tilde{h}} \in \mathbb{R}^{N^{\textrm{h}} \times 2U}$ and $\boldsymbol{U}^{\textrm{r}} \in \mathbb{R}^{N^{\textrm{h}} \times N^{\textrm{h}}}$ are the weight matrices of the hidden states, and $\odot$ is an element-wise multiplication. \\
The update gate $\boldsymbol{z}^\textrm{G}_{\tau}$ controls how much of the past information is retained and how much of the new input is incorporated into the current hidden states. The update of $\boldsymbol{z}^{\textrm{G}}_{\tau}$ is given by 
\begin{equation}\label{eq:update}
    \boldsymbol{z}^{\textrm{G}}_{\tau} = \phi \left( \boldsymbol{W}^{\textrm{z}} \overline{\boldsymbol{s}}_{\tau} + \boldsymbol{U}^{\textrm{z}} \boldsymbol{h}_{\tau-1} \right), 
\end{equation}
where $\boldsymbol{W}^{\textrm{z}} \in \mathbb{R}^{N^{\textrm{h}} \times 2U}$ and $\boldsymbol{U}^{\textrm{z}} \in \mathbb{R}^{N^{\textrm{h}} \times N^{\textrm{h}}}$ are the weight matrices of the update gate. Given (\ref{eq:reset}), (\ref{eq:candidate}), and (\ref{eq:update}), the hidden state $\boldsymbol{h}_{\tau}$ can be updated by 
\begin{equation}\label{eq:hidden}
    \boldsymbol{h}_{\tau} = \left( \mathbf{1} - \boldsymbol{z}^{\textrm{G}}_{\tau} \right) \odot \tilde{\boldsymbol{h}}_{\tau} + \boldsymbol{z}^{\textrm{G}}_{\tau} \odot \boldsymbol{h}_{\tau-1}. 
\end{equation}
By Given the hidden state $\boldsymbol{h}_t$, the output of the GRU model estimates the state of the physical network at time slot $t+1$ as follows: 
\begin{equation}\label{eq:s_pred}
    \hat{\boldsymbol{s}}_{t+1} = \boldsymbol{W}^{\textrm{o}} \boldsymbol{h}_t, 
\end{equation}
where $\boldsymbol{W}^{\textrm{o}}$ is the output weight matrix. 

\subsection{GRU Training}
The GRU model is used to approximate the relationship between the input $\overline{\boldsymbol{S}}_t$ and the output $\hat{\boldsymbol{s}}_{t+1}$. Hence, the loss function of the GRU-based estimation model is given by
\begin{equation}\label{eq:loss_gru}
    J_{G} = \frac{1}{2U} \lVert \hat{\boldsymbol{s}}_{t+1} - \boldsymbol{s}_{t+1} \rVert^2.   
\end{equation}
To train our proposed GRU-based estimation model, we use a mini-batch stochastic gradient decent (SGD) method to minimize (\ref{eq:loss_gru}). The update rule of the parameter matrices in the GRU model is defined as  
\begin{equation}
    \begin{split}
        \boldsymbol{W}^i \leftarrow \boldsymbol{W}^i - \lambda_G \nabla_{\boldsymbol{W}^i} J_{G}, \\
        \boldsymbol{U}^j \leftarrow \boldsymbol{U}^j - \lambda_G \nabla_{\boldsymbol{U}^j} J_{G}, 
    \end{split}
\end{equation}
where $i \in \{ \textrm{r}, \textrm{z}, \tilde{h}, \textrm{o} \}$, $j \in \{ \textrm{r}, \textrm{z}, \tilde{h} \}$, $\lambda_G$ is the learning rate, and $\nabla_{\boldsymbol{W}^i} J_{G}$ and $\nabla_{\boldsymbol{U}^j} J_{G}$ are the gradient of the loss function with respect to $\boldsymbol{W}^i$ and $\boldsymbol{U}^j$. \cite{9210812}

\subsection{Components of the VDN Algorithm}
Given the status predicted by the DNT, we use a VDN based method to solve problem (\ref{eq:problem}). Next, we first introduce the components of the VDN algorithm. The VDN algorithm consists of seven components: 1) agents, 2) states, 3) actions, 4) policy, 5) reward function, 6) local Q function, and 7) global Q function, which are introduced as follows \cite{10004943}: 

\subsubsection{Agents} The agents are $M$ BSs. Each BS implements an unique deep Q network (DQN), observes a part of the physical network, and made decisions for user-BS association and whether to transmit the partially observed physical network information to the cloud server. The DQNs of BSs are not controlled by any central nodes. 

\subsubsection{States} The local state of each BS describes the positions of users within its service area, while the global state of all BSs describes the positions of all users in the physical network. Hence, we use $\boldsymbol{s}^m_t$ to denote the local state of BS $m$ at time slot $t$, and $\boldsymbol{s}_t$ to denote the global state.

\subsubsection{Actions} For each BS $m$, the action describes: 1) whether BS $m$ transmits its partially observed physical network information to the cloud server at time slot $t$ and which RB will be used for the transmission, and 2) the subset of users that BS $m$ serves. Here, 1) is expressed by $| \boldsymbol{y}_{mt} |$, and 2) can be expressed by a vector $\boldsymbol{z}_{mt} = \left[ z_{mt,1}, ..., z_{mt,U} \right]$ with $z_{mt,u} \in \{ 0,1 \}$ indicating whether BS $m$ serves user $u$.
Hence, the action of BS $m$ at time slot $t$ is $\boldsymbol{a}^m_t = \left[ | \boldsymbol{y}_{mt} |, \boldsymbol{z}_{mt} \right]$, $\boldsymbol{a}^m_t \in \mathcal{A}$, where $\mathcal{A}$ is the action space. The joint action of all BSs at time slot $t$ is $\boldsymbol{a}_t = \left[ \boldsymbol{a}^1_t, ..., \boldsymbol{a}^M_t \right]$. 

\subsubsection{Policy} The policy $\boldsymbol{\pi}^m \left( \boldsymbol{a}^m_t | \boldsymbol{s}^m_t \right)$ of each BS $m$ is the conditional probability that BS $m$ chooses action $\boldsymbol{a}^m_t$ given its current state $\boldsymbol{s}^m_t$. 

\subsubsection{Reward Function} The team reward $r \left( \boldsymbol{s}_t, \boldsymbol{a}_t \right)$
captures the total reward of all BSs taking a joint action $\boldsymbol{a}_t$ under a global state $\boldsymbol{s}_t$. The team reward function of all BSs at each time slot $t$ is the weighted sum of: 1) the sum of data rates of all users, and 2) the similarity between the physical network status $\boldsymbol{s}_t$ and the DNT status $\overline{\boldsymbol{s}}_t$, which can be expressed by: 
\begin{equation}\label{eq:global_r}
     r \left( \boldsymbol{s}_t, \boldsymbol{a}_t \right) = \left \{
     \begin{alignedat}{2}
        &- \frac{1 - \epsilon}{U} \lVert \boldsymbol{s}_t - \overline{\boldsymbol{s}}_t \rVert_2^2 + \sum_{m=1}^M \sum_{u=1}^U \epsilon c_{umt}, \\
        & \hspace{13em} \text{if}~ \forall u,\xi_{ut} = 1, \\
        &\sum_{u=1}^U \mathbbm{1}_{\{ \xi_{ut} > 1 \}} \xi_{ut} \rho, \hspace{5.2em}\text{otherwise}. 
    \end{alignedat}
    \right.
\end{equation}
where $\xi_{ut} = \sum_{m=1}^M z_{mt,u}$ is the total number of BSs that serve user $u$ at time slot $t$, and $\rho < 0$ is a penalty for one user obtaining RBs from multiple BSs. From (\ref{eq:global_r}), we see that the sum of data rates of all users depends on the physical network information transmissions. The sum of data rates of all users cannot directly obtained using the joint action $\boldsymbol{a}_t$ since $\boldsymbol{a}^m_t$ only determines the users connect to BS $m$ (i.e., $\boldsymbol{z}_{mt}$) but does not consider how to allocate RBs to these users (i.e., $\boldsymbol{x}_{umt}$). Next, we introduce how to calculate the team reward of all BSs given an action $\boldsymbol{a}_t$. In particular, given an action $\boldsymbol{a}_t$, problem (\ref{eq:problem}) can be divided into $M$ suboptimization problems that can be solved by each BS iteratively. Hence, the suboptimization of BS $m$ is given by 
\begin{subequations}\label{eq:problem_2}
    \begin{equation}\tag{\theequation}\label{eq:sim_obj}
        \begin{split}
            \max_{\substack{\{ \boldsymbol{x}_{ijt} \}_{\left( i,j \right) \in \mathcal{Z}^m_t, t \in \mathcal{T}} \\ \{ \boldsymbol{y}_{kt} \}_{k \in \{ 1,...,m \}} }} & \sum_{t=1}^T \Bigg( - \frac{1 - \epsilon}{U} \lVert \boldsymbol{s}^m_t - \overline{\boldsymbol{s}}^m_t \rVert_2^2  \\
            & + \sum_{u \in \mathcal{U}^m_t} \epsilon c_{umt} \left( \boldsymbol{x}_{umt}, \left[ \boldsymbol{x}_{ijt} \right], \left[ \boldsymbol{y}_{kt} \right] \right) \Bigg), 
        \end{split}
    \end{equation}
    \begin{flalign}
        &&\text{s.t. } & x_{umt,n} \in \{ 0,1 \}, n \in \mathcal{N}, t \in \mathcal{T}, && \label{eq:sim_const1} \\
        && &\sum_{n=1}^N x_{umt,n} \leq 1, u \in \mathcal{U}^m_t, t \in \mathcal{T}, && \label{eq:sim_const2} \\
        && &\sum_{u \in \mathcal{U}^m_t} x_{umt,n} \leq 1, t \in \mathcal{T}, && \label{eq:sim_const3} \\
        && & \sum_{u \in \mathcal{U}^m_t} x_{umt,n} + y_{mt,n} \leq 1, n \in \mathcal{N}, t \in \mathcal{T}. && \label{eq:sim_const5}
    \end{flalign}
\end{subequations}
where $\mathcal{Z}^m_t = \{ \left( i,j \right) \mid z_{jt,i}=1, j < m\}$. 
Since the associated users $\boldsymbol{z}_{mt}$ of BS $m$ have determined, only the RB allocation vectors $ \boldsymbol{x}_{umt}$ need to be optimized. In (\ref{eq:problem_2}), since the objective function (\ref{eq:sim_obj}) is linear, the constraints are non-linear, and the optimization variables are integers, we can use an iterative Hungarian algorithm \cite{jonker1986improving} to solve the suboptimization problems for all BSs. Compared to standard convex optimization algorithms, using the Hungarian algorithm to solve problem (\ref{eq:problem_2}) does not require to calculate gradients of each variables nor dynamically adjusting the step size for convergence. 
To implement Hungarian algorithm for finding the optimal RB allocation for all BSs, each BS must transform its suboptimization problem (\ref{eq:problem_2}) into a bipartite graph matching problem. Hence, a bipartite graph $\mathcal{G}^m = \left< \mathcal{U}^m \times \mathcal{N}^m, \mathcal{E}^m \right>$ must be constructed for each BS $m$, where $\mathcal{U}^m$ is the set of vertices represent the users associated with BS $m$, $\mathcal{N}^m$ is the set of vertices represent the RBs that are not occupied by physical network information transmission such that can be allocated to each associated user of BS $m$, and $\mathcal{E}^m$ is the set of edges that connect vertices from $\mathcal{U}^m$ and $\mathcal{N}^m$. Let $e^m_{un} \in \mathcal{E}^m$ be the edge connecting vertex $u$ in $\mathcal{U}^m$ and vertex $n$ in $\mathcal{N}^m$ with $e^m_{un} \in \{ 0,1 \}$, where $e^m_{un} = 1$ indicates that BS $m$ allocates RB $n$ to user $u$ (i.e., $x_{umt,n} = 1$), and $e^m_{un} = 0$ otherwise. 
Let $\mathcal{X}^m_t$ be a subset of $\mathcal{E}^m$ at time slot $t$, in which two edges can neither share a common vertex in $\mathcal{N}^m$, nor in $\mathcal{U}^m$. Therefore, in $\mathcal{X}^m_t$, each RB $n$ can only be allocated to one associated user (constraint (\ref{eq:const3}) and (\ref{eq:const6}) are satisfied), and each user can only occupy one RB (constraint (\ref{eq:sim_const2}) is satisfied). The weight of edge $e^m_{un}$ is 
\begin{equation}\label{eq:weight}
    \psi^m_{un} = \left \{
     \begin{alignedat}{2}
        &c_{umt}, \hspace{3em} && \text{if}~ e^m_{ut} = 1, \\
        &0, && \text{otherwise}. 
    \end{alignedat}
    \right.
\end{equation}
Given the formulated bipartite matching problem of BS $m$, we can find an optimal RB allocation scheme for BS $m$ when the RB allocation schemes of other BSs are given. Hence, to optimize the RB allocation schemes of all BSs, one must iteratively update the RB allocation scheme of each BS using the Hungarian algorithm \cite{9210812}. Algorithm \ref{alg:algorithm1} summarizes the procedure of using the Hungarian algorithm to find RB allocation schemes for all BSs. 
\begin{algorithm}[!t]
    \small
    \caption{\small The Process of Using the Hungarian algorithm to find RB allocation schemes for all BSs}
    \label{alg:algorithm1}
    \begin{algorithmic}
        \FOR {each environment step $t$}
            \IF {$m = 1$}
                \STATE \textbf{BS 1: } 
                \STATE Use its Q network to determine $\boldsymbol{a}^1_t$. 
                \STATE Based on $\boldsymbol{a}^1_t$, use the Hungarian algorithm to determine the optimal $\mathcal{X}^{1*}_t$. 
                \STATE Use $\mathcal{X}^{1*}_t$ to find the optimal RB allocation policy $\left[ \boldsymbol{x}_{u1t} \right]_{u \in \mathcal{U}^1_t}$. 
                \STATE Send $\left[ \boldsymbol{x}_{u1t} \right]_{u \in \mathcal{U}^1_t}$ to BS 2. 
            \ELSE
                \FOR {$m = 2 \to M$}
                \STATE \textbf{BS $m$: }
                \STATE Use its Q network to determine $\boldsymbol{a}^m_t$. 
                \STATE Receive $\left[ \boldsymbol{x}_{ijt} \right]_{i,j \in \mathcal{Z}^m_t}$ from BS $m-1$. 
                \STATE Based on $\boldsymbol{a}^m_t$ and $\left[ \boldsymbol{x}_{ijt} \right]_{i,j \in \mathcal{Z}^m_t}$, use the Hungarian algorithm to determine the optimal $\mathcal{X}^{m*}_t$. 
                \STATE Use $\mathcal{X}^{m*}_t$ to find the optimal RB allocation policy $\left[ \boldsymbol{x}_{umt} \right]_{u \in \mathcal{U}^m_t}$. 
                \STATE Send $\left[ \boldsymbol{x}_{ijt} \right]_{i,j \in \mathcal{Z}^{m+1}_t}$ to BS $m+1$. 
                \ENDFOR
            \ENDIF
        \ENDFOR
    \end{algorithmic}
\end{algorithm}

\subsubsection{Local Q Function} The local Q function $Q_m \left( \boldsymbol{s}^m_t , \boldsymbol{a}^m_t \right)$ of BS $m$ is used to estimate the expected reward of the BS taking action $\boldsymbol{a}^m_t$ under a local state $\boldsymbol{s}^m_t$. 
Each BS uses a GRU, referred to as the Q network, parameterized by $\boldsymbol{\theta}_m$, to approximate its Q function. The output of the Q network represents the Q values under a given state and different actions. Therefore, the local Q function approximated by the Q network with parameters $\boldsymbol{\theta}_m$ is expressed as $Q_{\boldsymbol{\theta}_m} \left( \boldsymbol{s}^m_t , \boldsymbol{a}^m_t \right)$. 

\subsubsection{Global Q Function} The global Q function $Q \left( \boldsymbol{s}_t, \boldsymbol{a}_t \right)$ is defined as a Q function that estimates the reward of all BSs taking action $\boldsymbol{a}_t$ under a global state $\boldsymbol{s}_t$. In our algorithm, we cannot obtain the global Q function since we do not have a centralized DQN to approximate the global Q function. However, if we want the BSs to maximize the team reward, we need to use global Q value to update the DQN model of each BS. During the training process, we will introduce an approximate a global Q function $Q_{tot} \left( \boldsymbol{s}_t, \boldsymbol{a}_t \right)$ that is estimated using local Q values. 

\subsection{Training of the VDN}
Next, we introduce the training process of the VDN algorithm. Since we consider a MARL and each BS will not share its reward, actions, and states with other BSs, we cannot obtain the values of the global Q values. However, 
we need the global Q function 
to search for polices that optimize $r \left( \boldsymbol{s}_t, \boldsymbol{a}_t \right)$. To obtain the global Q values, we assume that a global Q value can be additively decomposed into local Q values of all BSs, which is given by
\begin{equation}\label{eq:global_q}
    Q_{tot} \left( \boldsymbol{s}_t , \boldsymbol{a}_t \right) = \sum_{m=1}^M Q_{\boldsymbol{\theta}_m} \left( \boldsymbol{s}^m_t , \boldsymbol{a}^m_t \right). 
\end{equation}
Using (\ref{eq:global_q}), each BS can update their Q networks in a distributed manner since the BSs do not need to know $Q \left( \boldsymbol{s}_t, \boldsymbol{a}_t \right)$, we introduce two key techniques to improve the training process of the designed MARL method. First, we use a neural network $\tilde{\boldsymbol{\theta}}_m$ whose structure is similar to that of Q network $\boldsymbol{\theta}_m$, to approximate the target Q function $\tilde{Q}_{\tilde{\boldsymbol{\theta}}_m
} \left( \boldsymbol{s}^m_t, \boldsymbol{a}^m_t \right)$ so as to define the loss function. Different from $\boldsymbol{\theta}_m$, the target Q network $\tilde{\boldsymbol{\theta}}_m$ updates and hence, the update speed of $\tilde{\boldsymbol{\theta}}_m$ is much slower compared to that of $\boldsymbol{\theta}_m$. For instance, the Q network $\boldsymbol{\theta}_m$ is updated at each training epoch while $\tilde{\boldsymbol{\theta}}_m$ is updated every several training epochs. 
The target value is expressed as 
\begin{equation}\label{eq:target}
    \tilde{y}_m = r^m \left( \boldsymbol{s}^m_t, \boldsymbol{a}^m_t \right) + \gamma \max_{\boldsymbol{a}^m_{t+1}} \tilde{Q}_{\tilde{\boldsymbol{\theta}}_m} \left( \boldsymbol{s}^m_{t+1}, \boldsymbol{a}^m_{t+1} \right), 
\end{equation}
where $\gamma$ is the discount factor. Similar to $Q \left( \boldsymbol{s}_t, \boldsymbol{a}_t \right)$, the target value of the global Q function is  
\begin{equation}
    \tilde{y} = r \left( \boldsymbol{s}_t, \boldsymbol{a}_t \right) + \gamma \max_{\boldsymbol{a}_{t+1}} \sum_{m=1}^M \tilde{Q}_{\tilde{\boldsymbol{\theta}}_m} \left( \boldsymbol{s}^m_{t+1}, \boldsymbol{a}^m_{t+1} \right). 
\end{equation}
Here, $\tilde{y}$ will be used to calculate the training loss of our designed VDN algorithm \cite{mnih2015human}. Second, we use experience replay technique to store past states, actions, and rewards in a buffer \cite{lin1992reinforcement}. The experience replay technique allows the model to learn from a diverse set of states, actions, and rewards sampled from the buffer, improving training stability and the performance of the model by breaking the correlation between consecutive states. To introduce the experience replay technique, we first define a transition as $\left( \boldsymbol{s}_t, \boldsymbol{a}_t, r \left( \boldsymbol{s}_t, \boldsymbol{a}_t \right), \boldsymbol{s}_{t+1} \right)$ and assume that each BS will collect $D$ transitions per training epoch $g$. These transitions are then stored in a replay memory, which also contains transitions collected from previous training epochs. In each training epoch $g$, a set $\mathcal{D}_g$ of transitions is sampled from the replay memory to update the local Q functions. Accordingly, the loss function of the VDN algorithm is defined as \cite{mnih2015human}
\begin{equation}\label{eq:GVDN}
\begin{split}
    J_Q & = \mathbb{E}_{\left( \boldsymbol{s}_t, \boldsymbol{a}_t, r, \boldsymbol{s}_{t+1} \right) \sim \mathcal{D}_g} \Bigg[ \Bigg( \Big( r \left( \boldsymbol{s}_t, \boldsymbol{a}_t \right) \\
    & \quad + \gamma \max_{\boldsymbol{a}_{t+1}} \sum_{m=1}^M \tilde{Q}_{\tilde{\boldsymbol{\theta}}_m} \left( \boldsymbol{s}^m_{t+1}, \boldsymbol{a}^m_{t+1} \right) \Big) - Q_{tot} \left( \boldsymbol{s}_t , \boldsymbol{a}_t \right) \Bigg)^2 \Bigg] \\
    & = \mathbb{E}_{\left( \boldsymbol{s}_t, \boldsymbol{a}_t, r, \boldsymbol{s}_{t+1} \right) \sim \mathcal{D}_g} \left[ \Big( \tilde{y} - Q_{tot} \left( \boldsymbol{s}_t , \boldsymbol{a}_t \right) \Big)^2 \right]. 
\end{split}
\end{equation}
Given the loss function, the parameters $\boldsymbol{\theta}_m$ of each BS $m$ is updated by a SGD method as follows: 
\begin{equation}\label{eq:network_param_update}
    \boldsymbol{\theta}_m \leftarrow \boldsymbol{\theta}_m - \lambda_Q \nabla_{\boldsymbol{\theta}_m} J_Q, 
\end{equation}
where $\lambda_Q$ is the learning rate, and $\nabla_{\boldsymbol{\theta}_m} J_Q$ is the gradient of $J_Q$ with respect to $\boldsymbol{\theta}_m$, which is given by 
\begin{equation}\label{eq:gradient}
\begin{split}
    &\nabla_{\boldsymbol{\theta}_m} J_Q \\
    & = \mathbb{E}_{\left( \boldsymbol{s}_t, \boldsymbol{a}_t, r, \boldsymbol{s}_{t+1} \right) \sim \mathcal{D}_g} \left[ \left( \tilde{y} - Q_{tot} \left( \boldsymbol{s}_t , \boldsymbol{a}_t \right) \right) \nabla_{\boldsymbol{\theta}_m} Q_{tot} \left( \boldsymbol{s}_t , \boldsymbol{a}_t \right) \right] \\
    & = \mathbb{E}_{\left( \boldsymbol{s}_t, \boldsymbol{a}_t, r, \boldsymbol{s}_{t+1} \right) \sim \mathcal{D}_g} \left[ \left( \tilde{y} - Q_{tot} \left( \boldsymbol{s}_t , \boldsymbol{a}_t \right) \right) \nabla_{\boldsymbol{\theta}_m} Q_{\boldsymbol{\theta}_m} \left( \boldsymbol{s}^m_t , \boldsymbol{a}^m_t \right) \right]. 
\end{split}
\end{equation}
The specific training procedure of the VDN algorithm is summarized in Algorithm \ref{alg:algorithm2}. 
\begin{algorithm}[!t]
    \small
    \caption{\small The Training Procedure if the VDN Algorithm }
    \label{alg:algorithm2}
    \begin{algorithmic}
        \REQUIRE Discount factor $\gamma$; learning rate $\lambda_Q$; training epoch $G$; environment steps per training epoch $D$; minibatch size $| \mathcal{D}_g |$.  
        \ENSURE Q network parameters $\boldsymbol{\theta}_1, ..., \boldsymbol{\theta}_M$ generated randomly.  
        \FOR {$g = 1 \to G$}
            \STATE \textbf{Data collection stage: }
            \FOR {each time step $t$}
                \STATE \underline{The DNT} estimates $\hat{\boldsymbol{s}}_t$ based on (\ref{eq:reset}) - (\ref{eq:s_pred}). 
                \FOR {\underline{each BS $m \in \mathcal{M}$}}
                    \STATE Record local state $\boldsymbol{s}^m_t$. 
                    \STATE Choose action $\boldsymbol{a}^m_t$ based on the current policy $\boldsymbol{\pi}^m \left( \boldsymbol{a}^m_t | \boldsymbol{s}^m_t \right)$. 
                    \STATE Determine the RB allocation scheme by using the Hungarian algorithm. 
                    \IF {$| \boldsymbol{y}_{mt} | = 1$}
                        \STATE Synchronize $\boldsymbol{s}^m_t$ with the DNT. 
                    \ENDIF
                \ENDFOR 
                \STATE \underline{The DNT} calculate $\overline{\boldsymbol{s}}_t$ based on (\ref{eq:s_dnt}).
                \STATE \underline{Each BS $m$} receives the team reward $r \left( \boldsymbol{s}_t, \boldsymbol{a}_t \right)$ based on (\ref{eq:global_r}). 
                \STATE \underline{Each BS $m$} stores the transition into the reply memory. 
            \ENDFOR
            \STATE \textbf{Training stage: }
            \FOR {\underline{each BS $m \in \mathcal{M}$} in parallel}
                \STATE Sample a set $\mathcal{D}_g$ of transitions from the reply memory. 
                \FOR {$d = 1 \to | \mathcal{D}_g |$}
                    \STATE Calculate the approximated global Q value $Q_{tot} \left( \boldsymbol{s}_d, \boldsymbol{a}_d \right)$ based on (\ref{eq:global_q}). 
                \ENDFOR
                \STATE Update $\boldsymbol{\theta}_1, ..., \boldsymbol{\theta}_M$ based on (\ref{eq:network_param_update}). 
            \ENDFOR
        \ENDFOR
    \end{algorithmic}
\end{algorithm}

\section{Complexity and Convergence of the Proposed Algorithm}\label{se:analysis}
In this section, we analyze the complexity and the convergence of the proposed method for DNT synchronization and user data rate optimization in the physical network. 


\subsection{Complexity Analysis}
We first analyze the complexity of our proposed method, which includes: 1) the DNT state predictions, 2) each BS determining its associated users and whether to send its partial observed physical network state to the DNT, and 3) RB allocation schemes. Thus, the complexity of the proposed method is determined by: 1) the complexity of the GRU-based predictive model, 2) the complexity of the Q networks of the VDN algorithm, and 3) the complexity of the Hungarian algorithm. Next, we will introduce the complexity of these components. 

Using the result from \cite{NIPS2017_3f5ee243, 9140367}, we know that the computational complexity of the GRU-based predictive model is $\mathcal{O} \left( K N^{\textrm{h}} \left( U + N^{\textrm{h}} \right) \right)$, and the computational complexity of each Q network at a BS is $\mathcal{O} \left( \theta^{\textrm{h}} \left( U + \theta^{\textrm{h}} + \left| \mathcal{A} \right| \right) \right)$, where $\theta^{\textrm{h}}$ is the size of the hidden states of the Q network. Hence, the complexity of all BS's Q networks is $\mathcal{O} \left( M \theta^{\textrm{h}} \left( U + \theta^{\textrm{h}} + \left| \mathcal{A} \right| \right) \right)$. Since the GRU-based predictive model is trained at the cloud server, and the Q networks are trained at BSs, both possessing sufficient computational power, the training overhead of these models can be ignored \cite{9140367}. 

We now analyze the computational complexity for each BS $m$ using the Hungarian algorithm to determine its RB allocation scheme. First, each BS $m$ requires $\left| \mathcal{U}^m \right| N$ iterations to calculate the data rate of each associated user over each RB. Then, the Hungarian algorithm updates $\psi^m_{un}$ to find the optimal matching set $\mathcal{X}^{m*}$. The worst-case complexity of each BS $m$ is $\mathcal{O} \left( \left| \mathcal{U}^m \right|^2 N \right)$, and the worst-case complexity of all BSs is $\mathcal{O} \left( \sum_{m=1}^M \left| \mathcal{U}^m \right|^2 N \right)$. In contrast, the best-case complexity of each BS $m$ is $\mathcal{O} \left( \left| \mathcal{U}^m \right| N \right)$, leading to a best-case complexity of $\mathcal{O} \left( \sum_{m=1}^M \left| \mathcal{U}^m \right| N \right) = \mathcal{O} \left( U N \right)$ for all BSs \cite{9210812}. 

Finally, we compare the complexity of the proposed method with the standard VDN algorithm that directly optimizes user association and RB allocation without using the Hungarian algorithm. The maximum complexity of our proposed method is 
\begin{equation}\label{eq:complexity}
    \mathcal{O} \Bigg( K N^{\textrm{h}} \left( U + N^{\textrm{h}} \right)  + M \theta^{\textrm{h}} \left( U + \theta^{\textrm{h}} + \left| \mathcal{A} \right| \right) + \sum_{m=1}^M \left| \mathcal{U}^m \right|^2 N \Bigg). 
\end{equation}
In contrast, the complexity of the standard VDN method is 
\begin{equation}\label{eq:complexity_b}
    \mathcal{O} \Big( K N^{\textrm{h}} \left( U + N^{\textrm{h}} \right)  + M \theta^{\textrm{h}} \left( U + \theta^{\textrm{h}} + \left| \mathcal{A}_b \right| \right) \Big), 
\end{equation}
where $\mathcal{A}_b$ is the action space of the standard VDN. In our proposed method, since the Q network of each BS is only used to determine the associated users and whether to send information to the DNT, we have $\left| \mathcal{A} \right| = 2^{U+1}$. However, in the standard VDN, each BS must consider to allocate $N$ RBs to $U$ users and the physical network information transmission link, resulting in an action space size of $\left| \mathcal{A}_b \right| = N^{U+1}$. Hence, our proposed method can significantly reduce the complexity of the standard method, especially in large networks with a big number $U$ of users. 

\subsection{Convergence Analysis}
Next, we analyze the convergence of the multi-agent VDN algorithm. 
We first define the gap between the actual global Q function $Q^{\pi} \left( \boldsymbol{s}, \boldsymbol{a} \right)$ and the global Q function $Q^{\pi}_{tot} \left( \boldsymbol{s}, \boldsymbol{a} \right)$ approximated by Q networks under a policy $\pi$ as 
\begin{equation}\label{eq:gap}
    \varepsilon^{\pi} \left( \boldsymbol{s}, \boldsymbol{a} \right) = Q^{\pi} \left( \boldsymbol{s}, \boldsymbol{a} \right) - Q^{\pi}_{tot} \left( \boldsymbol{s}, \boldsymbol{a} \right), 
\end{equation}
where
\begin{equation}\label{eq:q_opt}
    Q^{\pi} \left( \boldsymbol{s}, \boldsymbol{a} \right) = \sum_{\boldsymbol{s}'} P \left( \boldsymbol{s}' | \boldsymbol{s}, \boldsymbol{a} \right) \left[ r \left( \boldsymbol{s}, \boldsymbol{a} \right) + \gamma \max_{\boldsymbol{a}'} Q^{\pi} \left( \boldsymbol{s}', \boldsymbol{a}' \right) \right]
\end{equation}
with $\boldsymbol{s}', \boldsymbol{a}'$ being the global state and the joint action of all BSs at the next time step, and $P \left( \boldsymbol{s}' | \boldsymbol{s}, \boldsymbol{a} \right)$ being the transition probability from the current global state $\boldsymbol{s}$ to the next global state $\boldsymbol{s}'$ given the joint action $\boldsymbol{a}$. 

Then, the convergence of VDN is presented in the following lemma. 
\begin{lemma}\label{lm:1}
    \emph{If 1) $\varepsilon \to 0$, or $\gamma \to 1$, and 2) $ \left| \varepsilon^{\pi_1} \left( \boldsymbol{s}, \boldsymbol{a} \right) - \varepsilon^{\pi_2} \left( \boldsymbol{s}, \boldsymbol{a} \right) \right| \leq \varepsilon $ for any $\pi_1, \pi_2$,  
    our proposed VDN method can converge to the optimal $Q^*_{tot}$. }
\end{lemma}
\begin{proof}
    To analyze the convergence of the VDN, we first define 
    the Bellman operator $H^{\textrm{V}}$ of our proposed VDN method for updating the approximated global Q function $Q_{tot}$. Since $Q^{\pi}_{tot} \left( \boldsymbol{s}, \boldsymbol{a} \right) = Q^{\pi} \left( \boldsymbol{s}, \boldsymbol{a} \right) - \varepsilon^{\pi} \left( \boldsymbol{s}, \boldsymbol{a} \right)$ and the standard Bellman operator is $\left( H Q \right) \left( \boldsymbol{s}, \boldsymbol{a} \right) = \sum_{\boldsymbol{s}'} P \left( \boldsymbol{s}' | \boldsymbol{s}, \boldsymbol{a} \right) \left[ r \left( \boldsymbol{s}, \boldsymbol{a} \right) + \gamma \max_{\boldsymbol{a}'} Q \left( \boldsymbol{s}', \boldsymbol{a}' \right) \right]$, the Bellman operator of our proposed VDN method is \cite{fan2020theoretical, wang2022distributed}
    \begin{equation}\label{eq:vdn_operator}
        \begin{split}
            & \left( H^{\textrm{V}} Q^{\pi}_{tot} \right) \left( \boldsymbol{s}, \boldsymbol{a} \right) = \sum_{\boldsymbol{s}} P \left( \boldsymbol{s}' | \boldsymbol{s}, \boldsymbol{a} \right) \Big [ r \left( \boldsymbol{s}, \boldsymbol{a} \right) \\
            & \quad + \gamma \Big ( \max_{\boldsymbol{a}'} Q^{\pi}_{tot} \left( \boldsymbol{s}', \boldsymbol{a}' \right) + \varepsilon^{\pi} \left( \boldsymbol{s}', \boldsymbol{a}' \right) \Big ) \Big] - \varepsilon^{\pi} \left( \boldsymbol{s}, \boldsymbol{a} \right). 
        \end{split}
    \end{equation}
    Given (\ref{eq:vdn_operator}), we use the Banach fixed point theorem \cite{abounadi2002stochastic} to prove the convergence of our proposed VDN method. In particular, according to the Banach fixed point theorem, to prove the convergence of our proposed VDN method, we only need to prove that $H^\textrm{V}$ satisfies the following condition: 
    \begin{equation}\label{eq:contactor}
        \begin{split}
            \| H^\textrm{V} Q^{\pi_1}_{tot} - H^\textrm{V} Q^{\pi_2}_{tot} \|_{{\infty}} \leq \gamma \| Q^{\pi_1}_{tot} -  Q^{\pi_2}_{tot} &\|_{{\infty}}, \\ 
            &\text{for any } \pi_1, \pi_2,  
        \end{split}
    \end{equation}
    where $Q^{\pi_1}_{tot} \left( \boldsymbol{s}, \boldsymbol{a} \right)$ and $Q^{\pi_2}_{tot} \left( \boldsymbol{s}, \boldsymbol{a} \right)$ are the global Q function approximated by Q networks under policies $\pi_1$ and $\pi_2$, we have  
    \begin{equation}\label{eq:up_bound_1}
        \begin{split}
            &\| H^\textrm{V} Q^{\pi_1}_{tot} - H^\textrm{V} Q^{\pi_2}_{tot} \|_{{\infty}} \\
            & = \max_{\boldsymbol{s}, \boldsymbol{a}} \left| H^\textrm{V} Q^{\pi_1}_{tot} - H^\textrm{V} Q^{\pi_2}_{tot} \right| \\
            & = \max_{\boldsymbol{s}, \boldsymbol{a}} \Bigg | \sum_{\boldsymbol{s}'} P \left( \boldsymbol{s}' | \boldsymbol{s}, \boldsymbol{a} \right) \gamma \left[ \max_{\boldsymbol{a}_1'} \Big ( Q^{\pi_1}_{tot} \left( \boldsymbol{s}', \boldsymbol{a}_1' \right) + \varepsilon^{\pi_1} \left( \boldsymbol{s}', \boldsymbol{a}_1' \right) \Big ) \right.  \\
            & \quad \quad \quad \quad \left. - \max_{\boldsymbol{a}_2'} \Big ( Q^{\pi_2}_{tot} \left( \boldsymbol{s}', \boldsymbol{a}_2' \right) + \varepsilon^{\pi_2} \left( \boldsymbol{s}', \boldsymbol{a}_2' \right) \Big ) \right] \\
            & \quad \quad \quad \quad - \Big ( \varepsilon^{\pi_1} \left( \boldsymbol{s}, \boldsymbol{a} \right) - \varepsilon^{\pi_2} \left( \boldsymbol{s}, \boldsymbol{a} \right) \Big ) \Bigg | \\
            & = \max_{\boldsymbol{s}, \boldsymbol{a}} \! \Bigg | \! \sum_{\boldsymbol{s}'} \! P \left( \boldsymbol{s}' | \boldsymbol{s}, \boldsymbol{a} \right) \! \gamma \left( \max_{\boldsymbol{a}_1'} Q^{\pi_1} \! \left( \boldsymbol{s}', \boldsymbol{a}_1' \right) \! - \! \max_{\boldsymbol{a}_2'} Q^{\pi_2} \! \left( \boldsymbol{s}', \boldsymbol{a}_2' \right) \right) \\
            & \quad \quad \quad \quad - \Big ( \varepsilon^{\pi_1} \left( \boldsymbol{s}, \boldsymbol{a} \right) - \varepsilon^{\pi_2} \left( \boldsymbol{s}, \boldsymbol{a} \right) \Big ) \Bigg |. 
        \end{split}
    \end{equation}
    Since $\max_{\boldsymbol{a}_1'} Q^{\pi_1} \! \left( \boldsymbol{s}', \boldsymbol{a}_1' \right) \! - \! \max_{\boldsymbol{a}_2'} Q^{\pi_2} \! \left( \boldsymbol{s}', \boldsymbol{a}_2' \right) \leq \max_{\boldsymbol{a}'} \Big ( Q^{\pi_1} \! \left( \boldsymbol{s}', \boldsymbol{a}' \right) \! - \! Q^{\pi_2} \! \left( \boldsymbol{s}', \boldsymbol{a}' \right) \Big ) $, we have 
    \begin{equation}\label{eq:up_bound_2}
        \begin{split}
            &\| H^\textrm{V} Q^{\pi_1}_{tot} - H^\textrm{V} Q^{\pi_2}_{tot} \|_{{\infty}} \\
            & \leq \max_{\boldsymbol{s}, \boldsymbol{a}} \! \Bigg | \! \sum_{\boldsymbol{s}'} \! P \left( \boldsymbol{s}' | \boldsymbol{s}, \boldsymbol{a} \right) \! \gamma \max_{\boldsymbol{a}'} \Big ( Q^{\pi_1} \! \left( \boldsymbol{s}', \boldsymbol{a}' \right) \! - \! Q^{\pi_2} \! \left( \boldsymbol{s}', \boldsymbol{a}' \right) \Big ) \\
            & \quad \quad \quad \quad - \Big ( \varepsilon^{\pi_1} \left( \boldsymbol{s}, \boldsymbol{a} \right) - \varepsilon^{\pi_2} \left( \boldsymbol{s}, \boldsymbol{a} \right) \Big ) \Bigg |. 
        \end{split}
    \end{equation}
    Since $\sum_{\boldsymbol{s}'} \! P \left( \boldsymbol{s}' | \boldsymbol{s}, \boldsymbol{a} \right) \! \gamma \max_{\boldsymbol{a}'} \Big ( Q^{\pi_1} \! \left( \boldsymbol{s}', \boldsymbol{a}' \right) \! - \! Q^{\pi_2} \! \left( \boldsymbol{s}', \boldsymbol{a}' \right) \Big )$ is the expected value of $\gamma \max_{\boldsymbol{a}'} \Big ( Q^{\pi_1} \! \left( \boldsymbol{s}', \boldsymbol{a}' \right) \! - \! Q^{\pi_2} \! \left( \boldsymbol{s}', \boldsymbol{a}' \right) \Big )$
    with respect to $\boldsymbol{s}'$, it is less than or equal to $\max_{\boldsymbol{s}', \boldsymbol{a}'} \Big ( Q^{\pi_1} \! \left( \boldsymbol{s}', \boldsymbol{a}' \right) \! - \! Q^{\pi_2} \! \left( \boldsymbol{s}', \boldsymbol{a}' \right) \Big )$. Hence, we have 
    \begin{equation}\label{eq:up_bound_3}
        \begin{split}
            &\| H^\textrm{V} Q^{\pi_1}_{tot} - H^\textrm{V} Q^{\pi_2}_{tot} \|_{{\infty}} \\
            & \leq \max_{\boldsymbol{s}, \boldsymbol{a}} \Bigg | \gamma \max_{\boldsymbol{s}', \boldsymbol{a}'} \Big ( Q^{\pi_1} \! \left( \boldsymbol{s}', \boldsymbol{a}' \right) \! - \! Q^{\pi_2} \! \left( \boldsymbol{s}', \boldsymbol{a}' \right) \Big ) \\
            & \quad \quad \quad \quad - \Big ( \varepsilon^{\pi_1} \left( \boldsymbol{s}, \boldsymbol{a} \right) - \varepsilon^{\pi_2} \left( \boldsymbol{s}, \boldsymbol{a} \right) \Big ) \Bigg |.   
        \end{split}
    \end{equation}
    In (\ref{eq:up_bound_3}), since $\max_{\boldsymbol{s}', \boldsymbol{a}'} \Big ( Q^{\pi_1} \! \left( \boldsymbol{s}', \boldsymbol{a}' \right)\! - \! Q^{\pi_2} \! \left( \boldsymbol{s}', \boldsymbol{a}' \right) \Big ) = \max_{\boldsymbol{s}, \boldsymbol{a}} \Big ( Q^{\pi_1} \! \left( \boldsymbol{s}, \boldsymbol{a} \right)\! - \! Q^{\pi_2} \! \left( \boldsymbol{s}, \boldsymbol{a} \right) \Big )$, and $\varepsilon^{\pi_1} \! \left( \boldsymbol{s}, \boldsymbol{a} \right) - \varepsilon^{\pi_2} \! \left( \boldsymbol{s}, \boldsymbol{a} \right) \geq \min_{\boldsymbol{s}, \boldsymbol{a}} \Big ( \varepsilon^{\pi_1} \! \left( \boldsymbol{s}, \boldsymbol{a} \right) - \varepsilon^{\pi_2} \! \left( \boldsymbol{s}, \boldsymbol{a} \right) \Big )$, we have 
    \begin{equation}\label{eq:up_bound_4}
        \begin{split}
            & \gamma \max_{\boldsymbol{s}', \boldsymbol{a}'}  \Big ( Q^{\pi_1} \! \left( \boldsymbol{s}', \boldsymbol{a}' \right) \! - \! Q^{\pi_2} \! \left( \boldsymbol{s}', \boldsymbol{a}' \right) \Big ) \! - \! \Big ( \varepsilon^{\pi_1} \! \left( \boldsymbol{s}, \boldsymbol{a} \right) - \varepsilon^{\pi_2} \! \left( \boldsymbol{s}, \boldsymbol{a} \right) \Big ) \\
            & \leq \! \max_{\boldsymbol{s}, \boldsymbol{a}} \! \gamma \! \Big ( \! Q^{\pi_1} \! \left( \boldsymbol{s}, \boldsymbol{a} \right) \! - \! Q^{\pi_2} \! \left( \boldsymbol{s}, \boldsymbol{a} \right) \! \Big ) \! - \! \min_{\boldsymbol{s}, \boldsymbol{a}} \! \Big ( \! \varepsilon^{\pi_1} \! \left( \boldsymbol{s}, \boldsymbol{a} \right) \! - \! \varepsilon^{\pi_2} \! \left( \boldsymbol{s}, \boldsymbol{a} \right) \! \Big ) \\
            & = \max_{\boldsymbol{s}, \boldsymbol{a}} \! \left[ \! \gamma \! \Big ( Q^{\pi_1} \! \left( \boldsymbol{s}, \boldsymbol{a} \right) \! - \! Q^{\pi_2} \! \left( \boldsymbol{s}, \boldsymbol{a} \right) \! \Big ) \! - \! \Big ( \! \varepsilon^{\pi_1} \! \left( \boldsymbol{s}, \boldsymbol{a} \right) \! - \! \varepsilon^{\pi_2} \! \left( \boldsymbol{s}, \boldsymbol{a} \right) \! \Big ) \! \right]. 
        \end{split}
    \end{equation}
    Based on (\ref{eq:up_bound_4}), (\ref{eq:up_bound_3}) can be further simplified as  
    \begin{equation}\label{eq:up_bound_5}
        \begin{split}
            &\| H^\textrm{V} Q^{\pi_1}_{tot} - H^\textrm{V} Q^{\pi_2}_{tot} \|_{{\infty}} \\
            & \leq \Bigg | \max_{\boldsymbol{s}, \boldsymbol{a}} \left[ \gamma \Big ( Q^{\pi_1} \! \left( \boldsymbol{s}, \boldsymbol{a} \right) \! - \! Q^{\pi_2} \! \left( \boldsymbol{s}, \boldsymbol{a} \right) \Big ) \right. \\
            & \quad \quad \quad \quad \left. - \Big ( \varepsilon^{\pi_1} \left( \boldsymbol{s}, \boldsymbol{a} \right) - \varepsilon^{\pi_2} \left( \boldsymbol{s}, \boldsymbol{a} \right) \Big ) \right] \Bigg | \\ 
            & \leq \max_{\boldsymbol{s}, \boldsymbol{a}} \Bigg | \gamma \Big ( Q^{\pi_1} \! \left( \boldsymbol{s}, \boldsymbol{a} \right) \! - \! Q^{\pi_2} \! \left( \boldsymbol{s}, \boldsymbol{a} \right) \Big ) \\
            & \quad \quad \quad \quad - \Big ( \varepsilon^{\pi_1} \left( \boldsymbol{s}, \boldsymbol{a} \right) - \varepsilon^{\pi_2} \left( \boldsymbol{s}, \boldsymbol{a} \right) \Big ) \Bigg | \\
            & = \max_{\boldsymbol{s}, \boldsymbol{a}} \! \Bigg | \gamma \Big ( Q^{\pi_1} \! \left( \boldsymbol{s}, \boldsymbol{a} \right) \! - \! Q^{\pi_2} \! \left( \boldsymbol{s}, \boldsymbol{a} \right) \! - \! \varepsilon^{\pi_1} \! \left( \boldsymbol{s}, \boldsymbol{a} \right) \! + \! \varepsilon^{\pi_2} \! \left( \boldsymbol{s}, \boldsymbol{a} \right) \Big) \\
            & \quad \quad \quad - \! \left( 1 \! - \! \gamma \right) \Big ( \! \varepsilon^{\pi_1} \! \left( \boldsymbol{s}, \boldsymbol{a} \right) \! - \! \varepsilon^{\pi_2} \! \left( \boldsymbol{s}, \boldsymbol{a} \right) \Big ) \Bigg | \\
            & = \max_{\boldsymbol{s}, \boldsymbol{a}} \! \Bigg | \gamma \Big ( Q^{\pi_1}_{tot} \! \left( \boldsymbol{s}, \boldsymbol{a} \right) \! - \! Q^{\pi_2}_{tot} \! \left( \boldsymbol{s}, \boldsymbol{a} \right) \Big) \\
            & \quad \quad \quad - \! \left( 1 \! - \! \gamma \right) \Big ( \! \varepsilon^{\pi_1} \! \left( \boldsymbol{s}, \boldsymbol{a} \right) \! - \! \varepsilon^{\pi_2} \! \left( \boldsymbol{s}, \boldsymbol{a} \right) \Big) \Bigg |. 
        \end{split}
    \end{equation}
    Since $\left| \varepsilon^{\pi_1} \left( \boldsymbol{s}, \boldsymbol{a} \right) - \varepsilon^{\pi_2} \left( \boldsymbol{s}, \boldsymbol{a} \right) \right| \leq \varepsilon$, we have $ \varepsilon^{\pi_1} \left( \boldsymbol{s}, \boldsymbol{a} \right) - \varepsilon^{\pi_2} \left( \boldsymbol{s}, \boldsymbol{a} \right) \geq -\varepsilon$. Hence, (\ref{eq:up_bound_5}) can be rewritten as  
    \begin{equation}\label{eq:up_bound_6}
        \begin{split}
            &\| H^\textrm{V} Q^{\pi_1}_{tot} - H^\textrm{V} Q^{\pi_2}_{tot} \|_{{\infty}} \\
            & \leq \max_{\boldsymbol{s}, \boldsymbol{a}} \Bigg | \gamma \Big ( Q^{\pi_1}_{tot} \! \left( \boldsymbol{s}, \boldsymbol{a} \right) \! - \! Q^{\pi_2}_{tot} \! \left( \boldsymbol{s}, \boldsymbol{a} \right) \Big) -  \left( 1 \! - \! \gamma \right) \varepsilon \Bigg | \\
            & = \gamma \max_{\boldsymbol{s}, \boldsymbol{a}} \Big | Q^{\pi_1}_{tot} \! \left( \boldsymbol{s}, \boldsymbol{a} \right) \! - \! Q^{\pi_2}_{tot} \! \left( \boldsymbol{s}, \boldsymbol{a} \right) \Big | +  \left( 1 \! - \! \gamma \right) \varepsilon \\
            & = \gamma \| Q^{\pi_1}_{tot} -  Q^{\pi_2}_{tot} \|_{{\infty}} + \left( 1 \! - \! \gamma \right) \varepsilon. 
        \end{split}
    \end{equation}
    From (\ref{eq:up_bound_6}), we see that when $\varepsilon \to 0$ or $\gamma \to 1$, the VDN Bellman operator $H^\textrm{V}$ satisfies (\ref{eq:contactor}). 
    Based on the Banach fixed-point theorem, 
    the VDN in our proposed method will converge to $Q_{tot}^*$. This completes the proof. 
\end{proof}

\section{Simulation Results}\label{se:simulation}
In this section, we perform extensive simulations to evaluate the performance of our designed GRU and VDN based method. We first introduce the setup of the simulations. Then, we analyze the simulation results of our designed method.   

\subsection{System Setup}
For simulations, we consider a $300 \times 100$ network area. $3$ BSs are located at the coordinates $\left[ -100,0 \right]$, $\left[ 0,0 \right]$, and $\left[ 100,0 \right]$, respectively. Each BS has a coverage radius of $60$. The BSs in the system serve $U=12$ users. The cloud server that generates and controls the DNT is located at $\left[ 0,50 \right]$. 
The GRU model used to predict the DNT states is implemented on the cloud server. The cloud server will collect a dataset of $2,000$ trajectories from all $U$ users to train the GRU model, with each trajectory having $30$ steps. For comparison purposes, we consider an independent Q learning method in which the Q network of each BS $m$ is trained by its local Q values $Q_m \left( \boldsymbol{s}^m, \boldsymbol{a}^m \right)$ rather than the global Q values $Q_{tot} \left( \boldsymbol{s}, \boldsymbol{a} \right)$ as follows: 
\begin{equation}\label{eq:iql_update}
    \boldsymbol{\theta}_m \leftarrow \boldsymbol{\theta}_m - \lambda_Q \nabla_{\boldsymbol{\theta}_m} J^m_Q, 
\end{equation}
where $J_Q^m = \mathbb{E}_{\left( \boldsymbol{s}_t^m, \boldsymbol{a}_t^m, r^m, \boldsymbol{s}_{t+1}^m \right) \sim \mathcal{D}_g} \Big[ \Big( \Big( r^m \left( \boldsymbol{s}_t^m, \boldsymbol{a}_t^m \right) + \gamma \max_{\boldsymbol{a}_{t+1}^m} \tilde{Q}_{\tilde{\boldsymbol{\theta}}_m} \left( \boldsymbol{s}^m_{t+1}, \boldsymbol{a}^m_{t+1} \right) \Big) - Q_{\boldsymbol{\theta}_m} \left( \boldsymbol{s}_t^m , \boldsymbol{a}_t^m \right) \Big)^2 \Big]$, with $r^m \left( \boldsymbol{s}_t^m, \boldsymbol{a}_t^m \right)$ being the local reward of BS $m$. The baseline model parameters are similar to that of the designed method. Other parameters used in the simulations are listed in Table \ref{tab:sys_params}. 
\begin{table}[!t]
    \caption{System Parameters \cite{zhang2023optimization, hu2021rethinking}}
    \label{tab:sys_params}
    \centering
    \begin{tabular}{|c|c|c|c|}
        \hline
        \textbf{Parameter} & \textbf{Value} & \textbf{Parameter} & \textbf{Value} \\
        \hline
        $U$ & 10 & $M$ & 3 \\
        \hline
        $G$ & 75 & $B$ & 1 \\
        \hline 
        $P$ & 1 & $N_0$ & $1 \times 10^{-5}$ \\
        \hline
        $\theta^{\textrm{h}}$ & 128 & $\lambda_G$ & $1 \times 10^{-3}$ \\
        \hline
        $N^{\textrm{h}}$ & 128 & $K$ & 5 \\
        \hline
        $\gamma$ & 0.2 & $\lambda_{Q}$ & $1 \times 10^{-4}$ \\
        \hline
    \end{tabular}
\end{table}  

\subsection{Simulation Results and Analysis}
\begin{figure}[!t]
  \begin{center}
    \includegraphics[width=8cm]{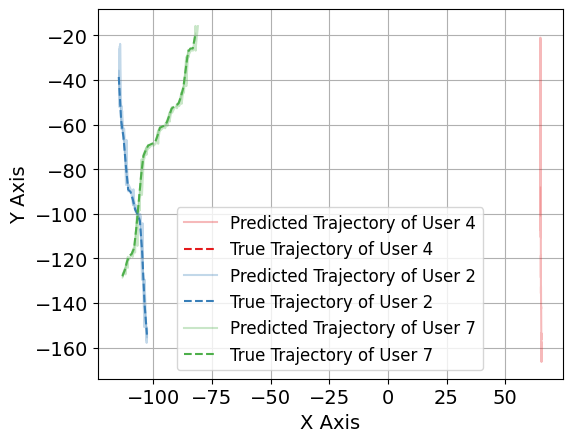}
    \caption{The prediction of the user movement trajectories via the GRU model. } 
    \label{fig:gru}
  \end{center}
\end{figure}
Fig. \ref{fig:gru} shows the user movement trajectories predicted by the GRU model. The users in Fig. \ref{fig:gru} are randomly selected from all users. From Fig. \ref{fig:gru}, we see that the positioning mean square errors of user 2, user 4, and user 7 are respectively 0.188, 0.075 and 0.320. This is because the GRU effectively captures dependencies of the historical user movement through its gating mechanisms. From this figure, we also see that the mobility prediction mean square errors of users 2 and 7 are higher compared to that of user 4. This is because the movement dynamics of users 2 and 7 are higher compared to that of user 4. 

\begin{figure}[!t]
  \begin{center}
    \includegraphics[width=8cm]{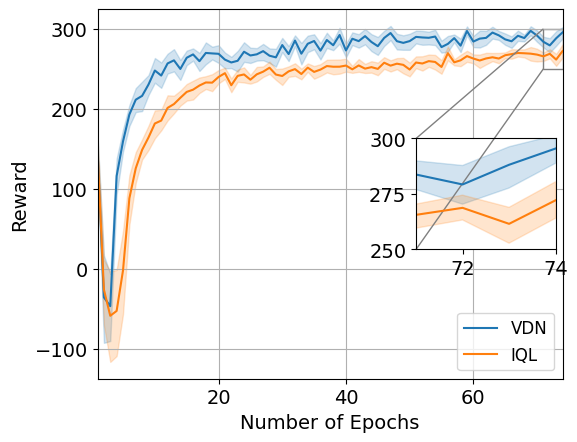}
    \caption{Convergence of VDN and IQL ($\epsilon = 0.3, N = 12$). } 
    \label{fig:vdn_vs_iql}
  \end{center}
\end{figure}
In Fig. \ref{fig:vdn_vs_iql}, we show the convergence of both the proposed method and the independent Q based method. Fig. \ref{fig:vdn_vs_iql} shows that, as the number of training epochs increases, the average rewards of both considered algorithms increase. This is because the policy of determining the associated users and the physical network information transmission is optimized by the considered RL algorithms. We also see that our designed algorithm can improve the weighted sum of data rates and the synchronization accuracy by up to 28.96\% compared to the independent Q. This gain stems from the fact that the proposed method allows each BS to use its local state to collaboratively find a globally optimal policy that maximize total data rates of all users, while minimize the asynchronization between the physical network and the DNT. 

\begin{figure}[!t]
  \begin{center}
    \includegraphics[width=8.5cm]{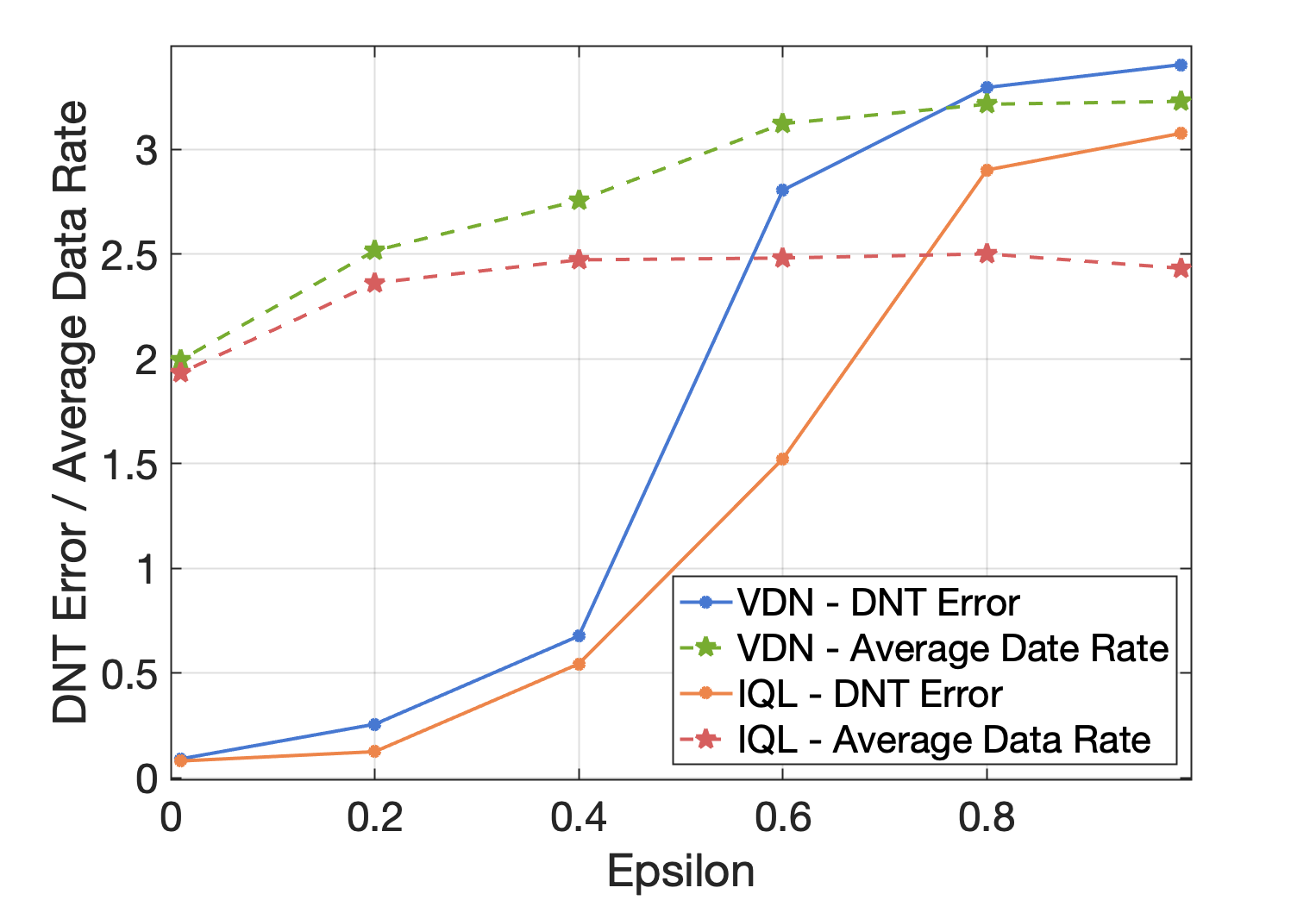}
    \caption{The average DNT error and the data rate as $\epsilon$ varies ($N = 9$). } 
    \label{fig:epsilon}
  \end{center}
\end{figure}
Fig. \ref{fig:epsilon} shows how the average DNT error and the average data rate of our proposed method and the independent Q method change as $\epsilon$ varies. Fig. \ref{fig:epsilon} shows that our proposed method improves the average data rate by up to $31.79\%$ compared to the independent Q method when $\epsilon = 0.8$, while the average DNT error of the proposed method is higher compared to the independent Q method. This is because our proposed method optimizes the weighted sum of the data rates of all users and the gap between the DNT and the physical network. 
Fig. \ref{fig:epsilon} also shows that the average data rate of the independent Q method reduces as $\epsilon$ increases from $0.8$ to $0.99$. This is because in the independent Q method, as the value of the epsilon increases, each BS prioritizes the data rate of its associated users, which may consequently increase the interference between BSs.

\begin{figure}[!t]
  \begin{center}
    \includegraphics[width=8.5cm]{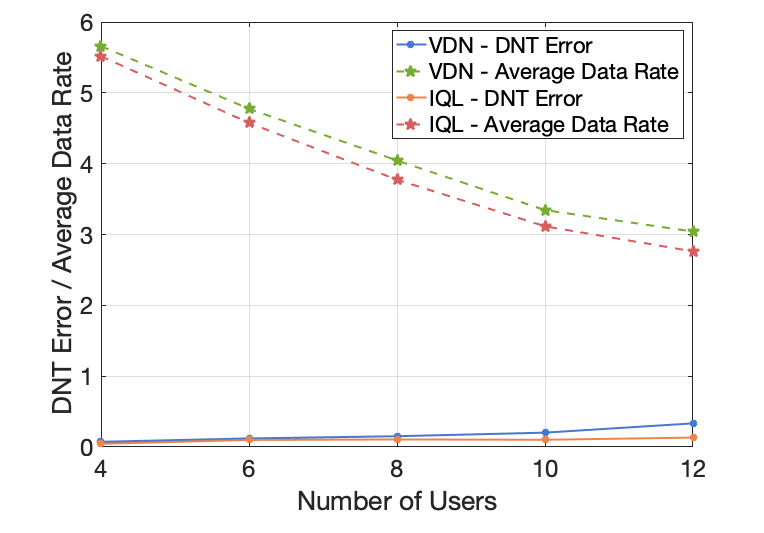}
    \caption{The average DNT error and the data rate as the number of users varies ($\epsilon = 0.25, N = 12$). } 
    \label{fig:user_num}
  \end{center}
\end{figure}
In Fig. \ref{fig:user_num}, we show how the average DNT error and the average data rate of our proposed method and the independent Q method change as the number of users varies. From Fig. \ref{fig:user_num}, we see that the DNT errors of both our considered algorithm and the independent Q method increase as the number of users increases. This is because as the network serves more users, the BSs may not have enough RBs to maintain DNT synchronization. Fig. \ref{fig:user_num} also shows that the average data rates of both methods decrease with the increase of the number of users. This stems from the fact that the BSs have limited RBs to serve users, such that several users are served by the RBs with large interference. Furthermore, Fig. \ref{fig:user_num} shows that when the number of users is $12$, the data rate of our proposed method is $10.15\%$ higher than that of the independent Q method but the average DNT error of the proposed method is $0.2$ higher compared to the independent Q method. This is because our proposed method adapts the RB allocation policy of each BS to accommodate the growing number of users to maintain a balance between serving the users in the physical network and the DNT synchronization.

\begin{figure}[!t]
  \begin{center}
    \includegraphics[width=8.5cm]{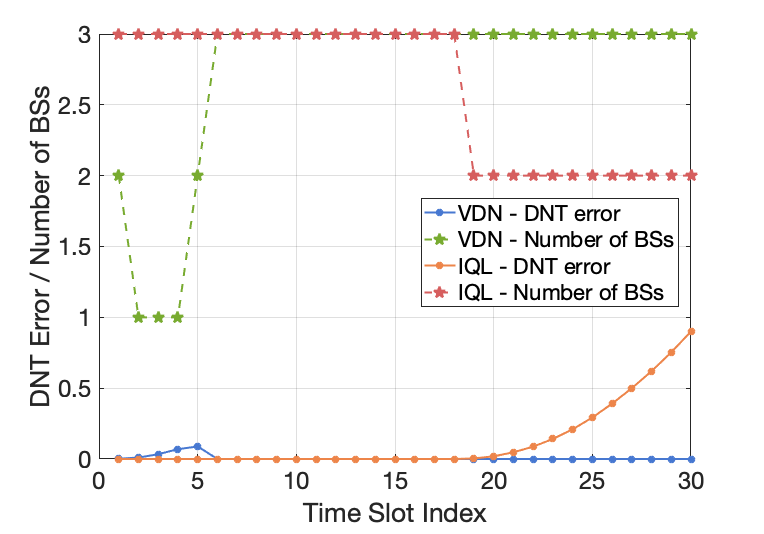}
    \caption{The number of BSs updating the DNT and the DNT errors at each step within a single episode ($\epsilon = 0.25, N = 12$). } 
    \label{fig:ep_detail}
  \end{center}
\end{figure}
Fig. \ref{fig:ep_detail} shows how the DNT errors of our proposed method and the independent Q method change as the number of BSs that updates the DNT at each step. 
In Fig. \ref{fig:ep_detail}, we show that our proposed method can achieve a small DNT error compared to the independent Q method. This is due to the fact that our method enables better BS collaborations, allowing them to optimize RB allocation by considering the joint impact of their decisions on the overall system performance. 
 
\section{Conclusion}\label{se:conclusion}
In this paper, we have proposed a DNT enabled network consisting of a physical network and its DNT. In this network, a set of BSs in the physical network must use their limited spectrum resources to serve a set of users while periodically transmitting the partial observed physical network information to a cloud server to update the DNT. We have formulated this resource allocation task as an optimization problem aimed at maximizing the sum of data rates for all users while minimizing the asynchronization between the physical network and the DNT. To address this problem, we have introduced a method based on GRUs and the VDN. The GRU component enables the DNT to predict its future state and maintain updates even when physical network information is not transmitted. The VDN component allows each BS to learn the relationship between its local observation and the team reward of all BSs, allowing it to collaborate with others in determining whether to transmit physical network information and optimizing spectrum allocation. Simulation results have demonstrated that, compared to a baseline approach utilizing GRU and IQL, our proposed method significantly improves the weighted sum of user data rates and the asynchronization between the physical network and the DNT. 


\bibliographystyle{IEEEbib}
\bibliography{references1}
\end{document}